\documentclass[11pt, doublespaced, draftcls, onecolumn]{IEEEtran}

\usepackage{subfigure}
\usepackage{epsfig}
\usepackage{dcolumn}
\usepackage[cmex10]{amsmath}
\usepackage{amsthm, amssymb, mathrsfs}
\usepackage{bbm,bm}
\usepackage{cite}
\usepackage{verbatim}
\usepackage{afterpage}
\usepackage{color}

\newtheorem{thm}{Theorem}[section]
\newtheorem{cor}{Corollary}[section]
\newtheorem{lem}[thm]{Lemma}
\newtheorem{defn}{Definition}[section]
\begin{document}
\date{}
\title{Quantum Internal Model Principle: Decoherence Control }
\author{Narayan~Ganesan,~\IEEEmembership{Member,~IEEE,}
        Tzyh-Jong~Tarn,~\IEEEmembership{Life Fellow,~IEEE}
\thanks{N. Ganesan is with Computer and Information Sciences Department, University of Delaware, email: ganesan@udel.edu}
\thanks{T. J. Tarn is with the Department of Electrical and Systems Engineering, Washington University in St. Louis, e-mail: tarn@wuauto.wustl.edu}
}


\maketitle

\begin{abstract}

In this article, we study the problem of designing a Decoherence Control for quantum systems with the help of a scalable ancillary quantum control and techniques from geometric control theory, in order to successfully and {\em completely} decouple an open quantum system from its environment. We re-formulate the problem of decoherence control as a disturbance rejection scheme which also leads us to the idea of {\em Internal Model Principle} for quantum control systems which is first of its kind in the literature.

It is shown that decoupling a quantum disturbance from an open quantum system, is possible only with the help of a quantum controller which takes into account the model of the environmental interaction. This is demonstrated for a simple 2-qubit system wherein the effects of decoherence are completely eliminated. The theory provides conditions to be imposed on the controller to ensure perfect decoupling. Hence the problem of decoherence control naturally gives rise to the quantum internal model principle which relates the disturbance rejecting control to the model of the environmental interaction.

Classical internal model principle and disturbance decoupling focus on different aspects viz. perfect output tracking and complete decoupling of output from external disturbances respectively. However for quantum systems, the two problems come together and merge in order to produce an effective platform for decoherence control. In this article we introduce a seminal connection between disturbance decoupling and the corresponding analog for internal model principle for quantum systems.
\end{abstract}

%
\IEEEpeerreviewmaketitle

\section{Introduction}
Quantum Information and Quantum Computation hold the key to faster information processing and better and reliable communication\cite{chuang}. The properties, the quantum superposition, coherence and entanglement are vital to quantum information processing. Quantum measurements in general collapses a quantum state into set of bases decided by the observable. Decoherence~\cite{mensky} is the process by which the quantum system loses the coherence and superposition by continually interacting with the environment. A quantum system that is continuously interacting with the environment is called an Open Quantum System~\cite{openqusys}. Decoherence is conceptually equivalent to a continuous and forcible collapse of the wave function of the system onto the basis decided by the environment(also called the {\em pointer basis}~\cite{zurek}). In practice, this {\em adiabatic} process takes a finite time in the order of a few milliseconds thus rendering the quantum system classical. The problem of decoherence is currently the biggest roadblock towards exploitation of quantum speedup in computation. Thus far, many researchers have proposed multitude of ways to control decoherence in such open quantum systems, of which a few representative contributions include open-loop pulses~\cite{viola1}\cite{viola2}\cite{proto}\cite{lidar}, and control within Decoherence Free Subspaces~\cite{lidar}. Open loop pulses techniques are designed for systems that are acted upon by pre-programmed control pulses. Such methods also have the tendency to annihilate useful non-predetermined control in addition to suppressing decoherence. This is debilitating for quantum information processing and do not work under arbitrary and undetermined control. Another class of ideas is based on the Decoherence Free Subspace(DFS)~\cite{lidar} which are proven immune to decoherence due to the degeneracy of the basis vectors with respect to the decoherence interaction. Such methods aim at encoding and steering the quantum information within such a subspace at all time. Such a strategy does not admit arbitrary control Hamiltonians as any transition out of the subspace would subject the state to decoherence and hence loss of information. 

Another class of work which is based on symptom or syndrome correction is error correction codes\cite{calderbank}\cite{shor}. These methods aim at correcting the observed effects or symptoms of decoherence. These methods usually require ancillary bits to encode a specific quantum information in a redundant fashion and perform posterior transformations depending on the observed error syndrome. Such methods require number of ancillary/redundant bits proportional to the size of the original system and might not be scalable in the long run.

Hence most of the proposed techniques are either ad-hoc, or limited in control functions or not scalable. In this work, we propose a scalable, strategy which preserves action of useful controls and lets the system evolve according to the same while eliminating the effects of decoherence. This is applicable to a wide class of control as well as decoherence Hamiltonians.

Our work is orthogonal to almost all published work in the literature on decoherence control/quantum disturbance decoupling. We strongly believe that our study provides another avenue of research for decoherence control, informs the readers in this field of new directions to be investigated for the same problem and lends a deeper insight to quantum disturbance decoupling. Needless to say, reviewing the extensive nature of the work in the literature on using density matrix approach to study the problem, one cannot expect to give a complete solution to this important problem in one paper. We present some pivotal and important results on decoherence control in this paper with the following main contributions:
\begin{enumerate}
 \item {\bf Open Loop Invariance:} Utilize differential geometric tools to perform structural analysis and extract important information regarding susceptibility of the given system to decoherence. This helps determine {\em a priori} whether or not the given system is immune to decoherence interaction and can be used to avoid going through, tedious, sometimes futile and time consuming work.
 \item {\bf Active Controller:} Provide results in terms of the given control equation and available control resources, whether complete decoupling of the effects of disturbances is possible, with the help of an active controller.
 \item {\bf Ancillary Quantum Control:} If the system is not decouplable, design an effective control system via an ancillary quantum system and an active controller, that achieves complete decoupling under {\em arbitrary} and non-predetermined control. To the best of our knowledge the decoherence prevention in the presence of arbitrary useful control has not been addressed before. Moreover our results also provide a systematic way(not {\em ad-hoc}) to construct the desired control. 
\end{enumerate}
To this effect, we first provide the criteria for any system to be naturally immune to decoherence in terms of the Lie Algebra of the operators involved, in the presence of arbitrary user generated control. The treatment is powerful and general enough to yield Decoherence Free Subspace(DFS) as a special case of the open loop control. In addition, this yields best ways to encode a given quantum information that is immune to decoherence under arbitrary control. Secondly, for those systems that are not immune to decoherence under arbitrary control, and systems undergoing decoherence, we employ an active controller. At this point, we transition from an operator algebra method to a vector field method on the tangent space of the manifold as this offers additional valuable insights into the geometric nature of the problem.  We present a scalable construction involving an ancillary system(single ancillary qubit for a finite number of system qubits) to achieve complete decoherence control. All of the analyses mentioned above are performed in the presence of arbitrary user generated control which preserves useful work while eliminating only the effects of decoherence.


Finally, we present the simulation results with the above control strategy. The above mentioned ideas come together in a coherent way into what is called "Quantum Internal Model principle`` wherein the model of interaction with the environment is indispensable to efficient disturbance decoupling which will be discussed in the last section of the paper.


 
\section{Mathematical Preliminaries}
A pioneering effort to study quantum control systems using bilinear input affine model was carried out by Huang et. al\cite{htc}. The model has since found various applications and is extremely useful in analyzing the controllability properties of a quantum system on the state space of analytic manifolds\cite{nelson} which draws upon the previous results on controllability of finite dimensional classical systems by Sussman and Jurdjevic\cite{sussman} which in turn follows the results by Kunita\cite{kunita1}\cite{kunita2} and Chow~\cite{chow}. In this section we explore the conditions for a scalar function represented by a quadratic form to be invariant under the dynamics of the above model(with the additional assumption of time-varying vector fields) in the presence of a perturbation or interaction Hamiltonian. Such a formalism can be seen to readily relate to decoherence in open quantum systems wherein a perturbation Hamiltonian that couples the system to the environment can be seen to play the role of {\it disturbance}. Classical disturbance decoupling~\cite{isidori} \cite{isidori1} \cite{isidori2} provides insightful results on eliminating the effects of disturbance from output, however it will also be seen that the aforementioned is not quite similar to quantum decoherence control problem and one should be extremely careful in adapting the classical results to quantum regime.

Let the quantum control system corresponding to an open quantum system\cite{openqusys} interacting with the environment~(figure~(\ref{opqusys})) be given by, 
\begin{eqnarray}
\frac{\partial\xi(t,x)}{\partial t}=&[H_0\otimes \mathcal{I}_e(t,x)+\mathcal{I}_s \otimes H_e(t,x)+H_{SE}(t,x) +\sum_{i=1}^{r}u_i(t)H_i\otimes \mathcal{I}_e(t,x)]\xi(t,x)
\label{opqusys}
\end{eqnarray}
where, $\mathcal{H}_s$ is the system's Hilbert space and $\mathcal{H}_e$ the environment's Hilbert space. $\mathcal{H}_s$ could be finite or infinite dimensional and $\mathcal{H}_e$ is generally infinite dimensional. $\xi(t,x)$ is the wave function of the system and environment. $H_0$ and $H_e$ are operators corresponding to the drift Hamiltonian of the system and environment while $H_i$'s correspond to the control Hamiltonian of the system. $H_{SE}$ governs the interaction between the system and the environment. The above operators are skew Hermitian and assumed to be time varying and dependent on the spatial variable. In addition to the above dynamical equation, we introduce a complex scalar functional, $y(t)$, as a bilinear form that carries information about the system,
\begin{equation}
y(t,\xi)=\langle \xi(t,x) | C(t,x) |\xi(t,x) \rangle \label{opeq}
\end{equation}
The bilinear form of the function resembles the expected value of an operator, but it does not necessarily correspond to a measurement output. The operator $C(t,x)$ (Hermitian or non-Hermitian) is a time-varying operator acting on system Hilbert space. In all the subsequent analysis we study the invariance properties of the above scalar map of the system acted upon by decoherence interaction. This could be thought of as a function that has to be regulated in the presence of controls and disturbance.
\begin{figure}
\begin{center}
\includegraphics[width=3.5in, height=1.9in]{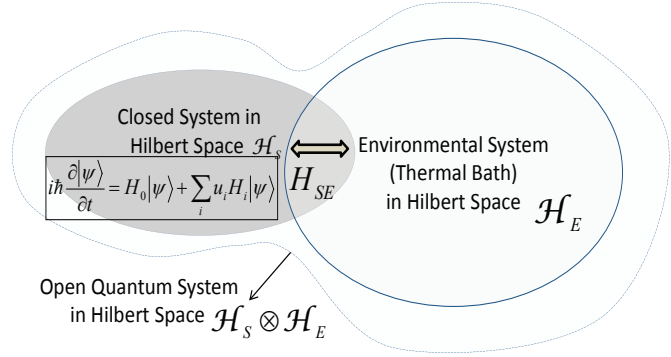}
\end{center}
\caption{An Open quantum system interacting with the environment via $H_{SE}$}
\end{figure}
\begin{defn}\label{invariance_definition}
Let $y(t,\xi)=f(t,x,u_1,\cdots,u_r,H_{SE})$ be a complex scalar map of the system as a function of the control functions and interaction Hamiltonian over a time interval $t_0\leq t\leq t_1$. The function is said to be invariant of the interaction Hamiltonian if
\begin{equation}
f(t,x,u_1,\cdots,u_r,H_{SE}) = f(t,x,u_1,\cdots,u_r,0) \label{cond9}
\end{equation}
for all admissible control functions $u_1,\cdots,u_r$ and a given interaction Hamiltonian $H_{SE}$.
\end{defn}

{\bf Forms of $y(t)$}. The above equation takes a quadratic form in the state $|\xi\rangle$ of the combined system and the environment. Some of the possible functions of interest, chosen for regulation are,

$(i)$ An expected value of a physical observable or an observation. The operator $C$, if Hermitian, can be a non-demolition observable and $y(t)$ is the output of the measurement performed on the system.

$(ii)$ By a suitable choice of the operator $C$ the value $y(t)$ can represent the coherence between the basis states of interest. It can be seen that a suitable value of the operator $C$ could yield the off-diagonal terms of the density matrix for the function $y(t)$. 

For example $C=|s_i\rangle \langle s_j| \otimes \mathbb{I}_e$ can be seen to yield the coherence between the orthogonal states of the system $|s_i\rangle$ and $|s_j\rangle$. For the pure state $\xi=\sum c_i |s_i\rangle$, $y(t) = c_i^* c_j$ and for the entangled  state $\xi = \sum c_i |s_i\rangle |e_i\rangle$ where $|e_i\rangle$ are the orthogonal states of the environment, a similar calculation yields $y=0$, denoting no coherence between the basis states.

$(iii)$ The operator $C$ could also be a linear combination of multiple non-Hermitian operators denoting coherence information of multiple basis vectors. This form for the operator $C$ is extremely useful in studying invariance of quantum information under external influences. This is presented in detail in the section~(\ref{calcs}) on DFS.

Let us define the corresponding free, control and the interaction vector fields as follows.
\begin{eqnarray}
K_0 &= (H_0 + H_e)|\xi\rangle\\
K_i &= H_i|\xi\rangle\\
K_I &= H_{SE}|\xi\rangle
\end{eqnarray}
Here we have suppressed the dependence of the Hamiltonians on time and spatial variable. As most of the practical systems are time-invariant and locality-invariant, this is a reasonable assumption. The following lemma\cite{ganesan} provides the basic conditions necessary for invariance of the scalar map with respect to the interaction vector field.
\begin{lem}
\label{opinvlem} Given that the quantum control system~(\ref{opqusys}) is analytic on the analytic manifold, the corresponding scalar map given by equation~(\ref{opeq}) is invariant under given $H_{SE}$(or the corresponding vector field, $K_I$) if and only if for all integers $p\geq0$ and any choice of vector fields $X_1,\cdots,X_p$ in the set $\{K_0,K_1,\cdots,K_r\}$,
\begin{equation}
L_{K_I} L_{X_1} \cdots L_{X_p}y(t,\xi) = 0; \mbox{for all }t,\xi \label{cond6}
\end{equation}
\end{lem}
Lemma~{\ref{opinvlem}} implies that the necessary and sufficient conditions for the scalar map $y$ of an analytic system to be invariant of the interaction vector field, $K_I$ are,
\begin{align}
L_{K_I} y(t,\xi) &= 0 \nonumber{} \\
L_{K_I} L_{K_{i_0}} \cdots L_{K_{i_n}} y(t,\xi) &= 0 \label{cond7}
\end{align}
for all $t, \xi$, for $0\leq i_0,\cdots, i_n \leq r$ and $n\geq 0$, where $K_0,\cdots, K_r$ are the corresponding drift and control vector fields and $K_I$, the interaction vector field.

\section{Invariance for the Quantum System}\label{calcs}
With the preceding mathematical preliminaries in place we can now apply the above conditions to the quantum system with careful consideration to the nature of the complex functional and the analytic manifold. In this section we present two different cases, (a) the open loop invariance, (b) open loop invariance under an external controller. Both these cases are studied for an open quantum system acted upon by arbitrary useful control.

\subsection{Open Loop Invariance with Arbitrary Control.}
We can now state the condition for invariance of the scalar function $y(t)$ with respect to a perturbation or interaction Hamiltonian, the proof and motivation for which is presented in\cite{ganesan}.

\begin{thm}\label{thm1}
Let $\mathcal{C}_0=C(t)$ and for $n=1,2,\cdots$, define
\begin{align*}
&\tilde{C}_n=\mbox{span}\{ad^j_{H_i}\mathcal{C}_{n-1}(t)|j=0,1,\ldots;i=1,\ldots
,r\}\\
&\mathcal{C}_n=\left\{ \left(ad_H+\frac{\partial}{\partial
t}\right)^j\tilde{C}_n; j=0,1,\cdots \right\}\\
&\vdots
\end{align*}
where $H=H_0+H_e$, the drift Hamiltonian of the combined system and environment, and $H_i, i=1\cdots, r,$ the control Hamiltonians. Define a distribution of operators, $\tilde{\mathcal{C}}(t)=\mbox{span}\{\mathcal{C}_1(t),\mathcal{C}_2(t),\cdots{}, \mathcal{C}_n(t),\cdots{}\}$. The scalar function represented by equation~(\ref{opeq}) of the quantum system is decoupled from the environmental interactions if and only if,
\begin{equation}
[\tilde{\mathcal{C}}(t), H_{SE}(t)]=0 \label{ic}
\end{equation}
\end{thm}
\begin{proof} The proof follows by noting the equivalence of equation~(\ref{cond7}) with the above condition. Consider the following term $L_{K_{i_0}} \cdots L_{K_{i_k}} y(x)$ for any $k \geq 1 $, and $i_0,\cdots, i_k \in \{0,\cdots r\}$. It can be expressed as a bilienar form of an operator of Lie brackets of $H_{i_0},H_{i_1},\cdots H_{i_r}, C $ and their time derivatives as follows. In particular for $k=0$,
\[
L_{K_{i_0}}y=\langle\xi|[C,H_{i_0}]+\delta(i_0)\frac{d}{dt}C|\xi\rangle \triangleq \langle\xi|T_1|\xi\rangle \\
\]
where $\delta(i_0)$ is the delta function that takes value $1$ when $i_0=0$ and the operator $T_1$ as defined above, is such that $T_1 \in \mathcal{C}_1$. Similarly for $k=1$ we have
\begin{align*}
L_{K_{i_1}}L_{K_{i_0}}y
&=\langle\xi|[[C,H_{i_0}],H_{i_1}]+[\delta(i_0)\frac{d}{dt}C,H_{i_1}] +\delta(i_1)\frac{d}{dt}([C,H_{i_0}] + \delta(i_0)\frac{d}{dt}C)|\xi\rangle \\
&\triangleq \langle\xi|T_2|\xi\rangle
\end{align*}
and $T_2 \in \mathcal{C}_2$. Continuing so, in general we have $T_n \in \mathcal{C}_n$. Now via condition~(\ref{cond7}), we have $\langle\xi |[H_{SE},T_n]|\xi\rangle = 0, \forall \xi$, or $[H_{SE},T_n]=0$ in general for invariance. Since the condition is true for any $n \geq 0$ and any $T_n$ and since the vector space of bounded linear operators is complete we have $[H_{SE},\sum_{i=0}^\infty \alpha_i T_i]=\sum_{i=0}^\infty \alpha_i[H_{SE},T_i]=0$ for $\alpha_i \in \mathbb{R}$. The converse is true by noting that any operator in the distribution $\mathcal{C}$ (i.e) for any $T \in \mathcal{C}$ can be decomposed into a sum of operators $\sum\alpha_i T_i$ for $T_i \in \mathcal{C}_i$ and given $[H_{SE},\sum_{i=0}^\infty \alpha_i T_i]=0 \forall \alpha_i$ which is true only when $[H_{SE},T_n] =0$ for any $n$. Hence from the previous equations $L_{K_I}L_{K_{i_n}}L_{K_{i_{n-1}}}\cdots L_{K_{i_0}}=0$ for $i_0,\cdots, i_k \in \{0,\cdots r\}$.
\end{proof}
We now present two qualitatively different examples, a system undergoing decoherence and a system that is immune to decoherence due to its Decoherence Free Subspace, to illustrate the applicability of the above open loop invariance theorem to practical quantum control systems.
\subsubsection{\bf Electro-optic Amplitude Modulation}
Consider a driven electromagnetic system in a single mode subject to decoherence. The control system describing the oscillator under semiclassical approximation is
\begin{align*}
\frac{d}{dt}\xi(t)=&(\omega a^{\dagger} a + \sum_j \omega_j b_j^\dagger b_j + i u(t)(a^{\dagger} -a) + a\sum_j g_j^* b_j + a^{\dagger}\sum_{j}g_j b_j^\dagger)\xi(t)
\end{align*}
where the system represented by mode $a$ is coupled to a bath of infinite number of oscillators, $b_j$ with corresponding coupling constants $g_j$ and where $\xi(t)$ is the combined wave function of the system and bath. Here $a, b_j$ and $a^\dagger, b_j^\dagger$, denote the photon creation and annihilation operators respectively. The control $u(t)$ is the strength of the input current and $H_0=\omega a^\dagger a + \sum_j \omega_j b_j^\dagger b_j$, $H_1=(a^\dagger -a)$ and $H_{SE}=a\sum_j g_j^* b_j + a^{\dagger}\sum_{j}g_j b_j^\dagger$ are the drift, control and decoherence Hamiltonians respectively. Let the system be monitored by a non-demolition observable\cite{barginsky},
\begin{align*}
C(t)=a\exp(i\omega t)+a^\dagger\exp(-i\omega t)
\end{align*}
with the corresponding output of the non-demolition measurement given by $y(t)=\langle\xi(t)|C(t)|\xi(t)\rangle$. Following theorem~\ref{thm1} we have $[C(t),H_1]=\mathrm{e}^{i\omega t} + \mathrm{e}^{-i\omega t} = 2\cos(\omega t)$ with vanishing higher order commutators. Hence $\tilde{C}_1=\{c_1 C+c_2 \mathbbm{I}\cos(\omega t), \forall c_1,c_2 \in \mathbbm{R}\}$ and since $[C(t),H_0] + {\partial C}/{\partial t} = 0$ we have $\mathcal{C}_1=\tilde{C}_1$. Since the commutant of the interaction Hamiltonian $H_{SE}=a\sum_{j}g_j^* b_j + a^{\dagger}\sum_{j}g_j b_j^\dagger$ with the elements of the set $\mathcal{C}_1$ not all zero, the non-demolition measurement is $(i)$ not invariant of the interaction Hamiltonian, $(ii)$ no longer back action evading due to the presence of the interaction. The measurement of the observable $C(t)$ would thus reveal information about the decoherence of the system.

\subsubsection{\bf Decoherence Free Subspaces(DFS) of a collection of 2-level systems:}\label{2levelsystems} The above theorem can also be applied to the problem of analyzing the DFS\cite{lidar}. Decoherence free subspaces camouflage themselves so as to be undetected by the interaction Hamiltonian due to degeneracy of their basis states with respect to $H_{SE}$ and the special algebraic properties of the interaction Hamiltonians. For a collection of 2-level systems interacting with a bath of oscillators the corresponding Hamiltonian is,
\[
H=\frac{\omega_0}{2}\sum_{j=1}^{N}\sigma_z^{(j)} + \sum_k \omega_k b_k^\dagger b_k + \sum_k\sum_{j=1}^{N} \sigma_z^{(j)}(g_{k}b_k^\dagger + g_{k}^*b_k)
\]
where the system is assumed to interact through the collective operator $\sum_j \sigma_z^{(j)}$ and $g_k$'s describe coupling to the mode $k$. An inquiry into what information about the system is preserved in the presence of the interaction can be answered by expressing the operator $C$ acting on the system Hilbert space in its general form in terms of the basis projection operators, 
\begin{align*}
C(t)=\sum_{i,j=0..2^N-1}c_{ij}|i\rangle\langle j|
\end{align*}
and solving for condition~(\ref{ic}). For a simple N=2 system, we have after straight forward calculations
\begin{align*}
\tilde{\mathcal{C}}=\mathrm{span}\{\sum_{i,j} c_{ij}|i\rangle\langle
j|.(j^{(1)}-i^{(1)}+j^{(2)}-i^{(2)})^K ,\forall K=0,1,2...\}
\end{align*}
where $j^{(l)}$ denotes the $l^{th}$ letter ($0$ or $1$) of the binary word $j$. Equation~(\ref{ic}) which is, $[\tilde{\mathcal{C}},H_{SE}]=0$ implies,
\begin{align*}
\sum_{i,j} c_{ij}|i\rangle\langle j|.(j^{(1)}-i^{(1)}+j^{(2)}-i^{(2)})^K =0, \forall K=1,2,3...
\end{align*}
or non-trivially $j^{(1)}+j^{(2)}=i^{(1)}+i^{(2)}$, or the two words have equal number of $1's$. The above calculations are valid for any finite $N$, a specific example for $N=3$ is $C=|000\rangle\langle 000|+|001\rangle\langle 001|+|010 \rangle\langle 100|+|011\rangle\langle 101|$. Of particular interest are terms such as, $|011\rangle\langle 101|$ and $|010\rangle\langle 100|$ as the corresponding $y(t)=\langle\xi(t)|C(t)|\xi(t)\rangle$ which is a function of the coherence between the basis states $|011\rangle, |101\rangle$ and $|010\rangle,|100\rangle$ is predicted to be invariant under the interaction. It is worth noting that the operator $C(t)$ acting on system Hilbert space here need not necessarily be Hermitian and only describes the quantum information that is preserved.

{\em Decoherence in the presence of control:} In the presence of the external controls $H_i= u_i\sigma_x^{(i)}$, the invariance condition~(\ref{ic}) is no longer satisfied for the operator $C=\sum_{i,j=0,..,2^N-1} |i\rangle\langle j|, i\neq j$ as $[[C,\sigma_x^{(i)}],\sigma_z^{(j)}]\neq 0$ and hence the coherence between the states $|i\rangle, |j\rangle$ is not preserved. This is because of the transitions outside the DFS caused by the control Hamiltonian. The above formalism is helpful in analyzing in general, class of information that would be preserved in the presence of interaction Hamiltonian which in turn could tell us about how to store information reliably in a quantum register in the presence of decoherence. Hence, in contrast to passive decoherence avoidance in the absence of external controls, this approach can  be used to determine the prudent means to {\em encode} quantum information, that stays immune to the decoherence interaction, even in the presence of arbitrary controls.

In summary, the notion of open loop invariance, (a) naturally gives rise to DFS described by the operator $C$, and helps perform extended analysis on the same, (b) was used to determine if a given scalar function(in this case non-demolition measurement) was affected by decoherence, and (c) could be used to design an operator $C$, which under the given system Hamiltonians and decoherence interaction would generate an invariant scalar map.
\subsection{Controller: $u = \alpha(\xi) + \beta(\xi)v$}
In this section we consider the case where the scalar map is not invariant in the presence of decoherence control(for eg. the coherence between states $|10\rangle$ and $|01\rangle$ of a 2 qubit system) and investigate the role of an active controller in order to decouple the same. From this point onward(and all the subsequent sections) we assume that the scalar map is not invariant, and that the decoherence operator($H_{SE}$) affects the scalar map. We study the effectiveness of an external controller whose control is dependent on the state of the system, in order to achieve invariance of the scalar map. The controller is assumed to be of the form, $u = \alpha(\xi) + \beta(\xi)v$ where the $u$ is implemented as a transformation involving matrices $\alpha$ and $\beta$. Here $v$ is $1\times r$, $u$ is $1\times r$, $\alpha(\xi)$, is a $1\times r$ vector and $\beta(\xi)$ is a non-singular matrix of of size $r\times r$(where $r$ is the number of open-loop controls). 

The realizability of the controller and implementation of the obtained control is a part of the ongoing work. It could potentially be realized with the help of a quantum machine(coherent control) or via quantum measurement/estimation theory. As outlined earlier, the focus of the current work is geometric analysis and a method for decoherence control design. The realization of the controller is an important open problem currently under study which could yield a few dissertations by itself.


The above form of an external control($u = \alpha(\xi) + \beta(\xi)v$) is general enough to encompass all popular control strategies as well as preserve the input-affine form of the original quantum control system even after the application of control. Consider the following system that is acted upon by the above controller of the form $u=\alpha(\xi)+\beta(\xi)v$,
\begin{equation}
\frac{\partial}{\partial t} \xi(t,x) = (H_0+H_e+\sum \alpha_i H_i)\xi(t) +
\sum_{i=1}^r v_i \sum_{j=1}^r \beta_{ij}H_j\xi(t) + H_{SE}\xi(t)\label{fbsys}
\end{equation}
where the new drift vector field, $\tilde{K}_0=(H_0+H_e+\sum \alpha_i H_i)\xi(t)$, control vector fields $\tilde{K}_i=\sum_j
\beta_{ij}H_j\xi(t)$, and decoherence interaction $K_I=H_{SE}\xi(t)$, are identified for the controlled system.

As stated earlier, the necessary and sufficient conditions for a scalar function $y(t)$ of the system to be invariant of the interaction vector field are,
\begin{align}
L_{K_I} y(t) &= 0 \nonumber{} \\
L_{K_I} L_{\tilde{K}_{i_0}} \cdots L_{\tilde{K}_{i_n}} y(t) &= 0 \label{cond7fb}
\end{align}
for $0\leq i_0,\cdots, i_n \leq r$ and $n\geq 0$. The above conditions can be stated in the form of the following result, which states it in terms of the operators defined in theorem(\ref{thm1}).
\begin{defn}\label{decouplable_definition}
The scalar map~(\ref{opeq}) of the system~(\ref{opqusys}) is said to be decouplable if there exist control parameters $\alpha$ and $\beta$, such that, under the corresponding controlled system~(\ref{fbsys}), the scalar map~(\ref{opeq}) is invariant in sense of
definition~(\ref{invariance_definition}).
\end{defn}
We now state the condition for the decouplability.
\begin{lem}\label{opralg_controller}
For the scalar map~(\ref{opeq}) of the quantum system~(\ref{opqusys}), is decouplable only if,
\begin{align*}
&[C,H_{SE}]=0\\
&[\tilde{\mathcal{C}}(t), H_{SE}(t)]\subset \tilde{\mathcal{C}}(t)
\end{align*}
where the distribution $\tilde{\mathcal{C}}(t)$ is as defined in theorem(~\ref{thm1}).
\end{lem}
\begin{proof} The above lemma only provides the necessary condition for the invariance. In the equations below we suppress the summation symbol and follow Einstein's convention, wherein a summation has to be assumed whenever a pair of the same index appears. Expanding out the corresponding terms of the equations~(\ref{cond7fb}), where the following equalities must hold for all $\xi$,
\begin{align*}
&L_{K_I}y=\langle\xi|[C,H_{SE}]|\xi\rangle=0\\
&L_{K_I}L_{\tilde{K}_i}y=\langle\xi|[[C,\beta_{ij}H_j],H_{SE}]+[C,H_j]L_{K_I}
\beta_{ij}|\xi\rangle=0\\
&L_{\tilde{K}_i}L_{\tilde{K}_0}y=\langle\xi|[\dot{C},\beta_{il}H_l] +[[C,H+\alpha_j H_j], \beta_{il}H_l] +[C,H_j]L_{\tilde{K}_i}\alpha_j|\xi\rangle=0
\end{align*}
\begin{align}
&L_{K_I}L_{\tilde{K}_i}L_{\tilde{K}_0}y \nonumber \\
=&\langle\xi|[[\dot{C},\beta_{il}H_l],H_{SE}]+[[C,H_j]L_{\tilde{K}_i}\alpha_j,H_{SE}] + [[[C,H+\alpha_j H_j], \beta_{il}H_l],H_{SE}] \nonumber \\
&+[\dot{C},H_l]L_{K_I}\beta_{il}+[C,H_j]L_{K_I}L_{\tilde{K}_i}\alpha_j + [[C,H],H_l]L_{K_I}\beta_{il} + [[C,H_j],H_l]L_{K_I}\alpha_j\beta_{il}|\xi\rangle \nonumber \\
=&0 \label{egeq}
\end{align}
The last equation above, provides a set of simultaneous equations to solve for the control parameters $\alpha$ and $\beta$ of the active controller, in order to achieve invariance. The above equation contains two types of terms. The terms containing $H_{SE}$ and terms that do not. The terms whose commutator with $H_{SE}$ is computed, is found to belong to the distribution $[\tilde{\mathcal{C}}(t),H_{SE}]$ whereas the terms without $H_{SE}$ is seen to belong to $\tilde{\mathcal{C}}(t)$. The above calculation can be extended to finite number of terms to arrive at the result. In order for the above equality to hold, in general one finds that, the condition for decouplability is relaxed to,
\begin{equation}
[\tilde{\mathcal{C}}(t),H_{SE}]\subset\tilde{\mathcal{C}}(t)(\mbox{compare to eq. \ref{ic}}) \label{condfb}
\end{equation}
\end{proof}
However, in order to solve equation~(\ref{egeq}) and consequently equation~(\ref{condfb}) for the control parameters($\alpha$ and $\beta$), it is important to study the properties of the operators in $\tilde{\mathcal{C}}$. It can be seen that the distribution $\tilde{\mathcal{C}}$, is generated by operators acting only on system Hilbert Space ($C, H_0, H_1 \cdots H_r$), whereas the operator $H_{SE}$ acts on the joint, system + environment Hilbert Space. Therefore, the above equation~(\ref{condfb}) cannot be solved for the control parameters unless $[\tilde{\mathcal{C}}(t),H_{SE}] = 0$. This is same as the open loop invariance without the active controller, which implies that as long as $H_{SE}$ and $H_0,\cdots, H_r$, act on different Hilbert spaces, the controller cannot act as an effective tool in decoupling the system.

Alternatively, for the controller to be an effective tool in solving the decoherence problem, the control Hamiltonians $H_i$'s have to act non-trivially on both the Hilbert spaces which would enable all the operators in equation~(\ref{condfb}) to act on the joint  system-environment Hilbert space.

In the rest of the paper we will outline a construction, involving an ancillary system and the active controller of the form $u=\alpha(\xi)+\beta(\xi)v$ in order to decouple and achieve complete invariance. We will revisit the 1 and 2-qubit systems
and present the applicability of the construction in achieving the final goal of decoherence control.

The operator algebra method outlined above was helpful in arriving at the invariance condition for open loop, with and without the active  controller. However, it only provides the necessary condition in order to be able to achieve invariance under the action of the controller.
Hence at this point we resort to an alternative approach to analyze the same problem, via the {\em invariant} subspace within the tangent space of the analytic manifold. In contrast to the operator algebra approach, which was based on the operators $H_0, H_1, \cdots, H_r$ and $H_{SE}$, we now use the control vector fields $K_1,\cdots, K_r$, and the decoherence vector field $K_I$, in order to analyze the invariance of the function, $y$ and leverage the geometry of system on the analytic manifold.

\section{Invariant Subspace Formalism}
In this section we present a alternate formalism to analyze the invariance of the function $y$.
\begin{defn}\label{obs}
Any vector field $K_\tau = K_I$ satisfying equations~(\ref{cond7}) is said to be in the orthogonal subspace of the observation space spanned by the one-forms,
\begin{align}
\mathcal{O} \triangleq \mbox{span}\{&dy(t,\xi),dL_{K_{i_0}}y(t,\xi),\cdots, dL_{K_{i_0}} \cdots L_{K_{i_n}} y(t,\xi), \cdots \} \nonumber\\
&\forall 0\leq i_0,\cdots, i_n \leq r \mbox{ and } n\geq 0
\end{align}
Denoted by $K_\tau \in \mathcal{O}^\perp$
\end{defn}
\begin{lem}
The distribution $\mathcal{O}^\perp$ is invariant with respect to the vector fields $K_0,\cdots,K_r$ under the Lie bracket operation. i.e., if $K_\tau \in \mathcal{O}^\perp$, then $[K_\tau,K_i] \in \mathcal{O}^\perp$ for $i=0,\cdots, r$
\end{lem}
Equations~(\ref{cond7}) after subtraction imply $L_{K_0}L_{K_I}y(t) - L_{K_I}L_{K_0}y(t)= L_{[K_0,K_I]}y(t)=0$. Similarly it is possible to derive other necessary conditions viz. $L_{[K_0,K_I]}L_{K_j}y(t)=0$ and $L_{K_j}L_{[K_0,K_I]}y(t)=0$ for invariance\cite{ganesan} which in turn imply that $L_{[[K_0,K_I],K_j]}y(t)=0$. In fact the above pattern of equations can be extended to any number of finite Lie brackets to conclude that,
\begin{equation}
L_{[[\cdots[K_I,K_{i_1}],K_{i_2}]\cdots K_{i_k}]}y(t)=0
\end{equation}
$1\leq i_1,i_2,\cdots,i_k \leq r$, which leads us to the definition of an {\em invariant} distribution $\Delta$ of vector fields with the following properties,
\begin{eqnarray}
K_\nu \in \Delta \implies L_{K_\nu}y(t)=0\\
K_\nu, K_\mu \in \Delta \implies [K_\mu, K_\nu] \in \Delta
\end{eqnarray}
and for any control/drift vector field $K_0, K_1,\cdots, K_r$,
\begin{equation}
K_\nu \in \Delta \implies [K_\nu,K_i]\in \Delta,\forall i\in{0,\cdots,r}
\end{equation}
This distribution is {\em involutive} and it is also observed(from the definition) that $K_I \in \Delta$. Such a distribution $\Delta$ is contained within $\mbox{ker}(dy(t,\xi))$. Hence $K_I\in \Delta \subset \mbox{ker}(dy)$. From the necessary conditions listed above the distribution is {\em invariant} under the control and drift vector fields $K_0,\cdots, K_r$. Simply stated,
\begin{equation}
[\Delta, K_i]\subset \Delta, \forall i\in{0,\cdots,r}\label{invariantwrtki}
\end{equation}
\begin{figure}
\begin{center}
\includegraphics[width=3.25in, height=2.0in]{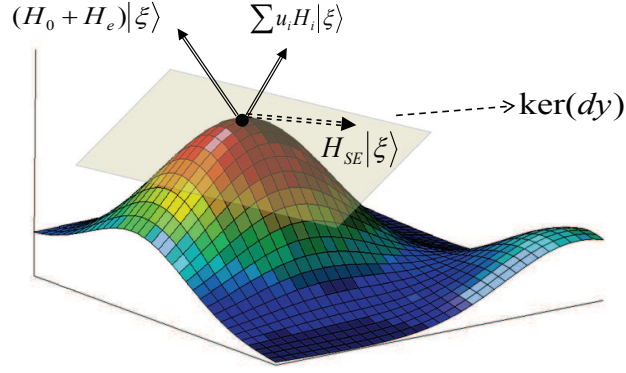}
\end{center}
\caption{The surface represents equal value of $y(t)$ for different values of $|\xi\rangle$ and the corresponding vector fields at the point $|\xi\rangle$ is denoted by the arrows and nullspace $\mbox{ker}(dy)$ is the tangent to the surface at the point $|\xi\rangle$. The dotted arrows lie in $\mbox{ker}(dy)$ and the solid arrows need not necessarily be contained in the same. The necessary condition for open loop invariance requires that $H_{SE}|\xi\rangle\in \mbox{ker}(dy)$.}\label{kerdy}
\end{figure}
It is also to be noted that the above calculations are reversible and the original necessary and sufficient conditions~(\ref{cond7}) can be derived starting from the invariant distribution. Hence the necessary and sufficient
conditions for open loop decouplability can now be restated (without proof) in terms of the invariant distribution\cite{ganesan2}.
\begin{thm}\label{invsp}
The scalar map $y(t)$ is unaffected by the interaction vector field $K_I$ if and only if there exists a distribution $\Delta$ with the following properties,\\
(i) $\Delta$ is invariant under the vector fields $K_0,K_1,\cdots, K_r$ i.e,\\
\begin{equation}
[\Delta, K_i]\subset \Delta, \forall i\in{0,\cdots,r}
\label{InvariantDistribution}
\end{equation}
(ii) $K_I\in \Delta \subset \mbox{ker}(dy(t))$
\end{thm}
A geometric representation of $\mbox{ker}(dy)$ is illustrated in figure~(\ref{kerdy}). The existence of the invariant subspace $\Delta$ is essential to decouplability of the given system. We now analyze the decouplability with the help of the active controller.
\subsection{Active Controller: $u =  \alpha(\xi) + \beta(\xi)v$}
In this section, we study the synthesis of the control parameters in the fundamental limit, that ensures complete decoupling from $H_{SE}$.
\begin{defn}
A distribution $\Delta$ is said to be controlled invariant on the analytic manifold $D_\omega$ if there exists a control pair ($\alpha, \beta$), $\alpha$, vector valued and $\beta$, matrix valued functions such that
\begin{align}
[\tilde{K}_0,\Delta ](\xi) \subset \Delta (\xi) \label{cid1}\\
[\tilde{K}_i,\Delta ](\xi) \subset \Delta (\xi) \label{cid2}
\end{align}
\[
\mbox{where, }\tilde{K}_0 = K_0 + \sum_{j=1}^r \alpha_j K_j \mbox{ and } \tilde{K}_i =
\sum_{j=1}^r \beta_{ij}K_j
\]
\end{defn}
The vector fields $\tilde{K}_0$ and $\tilde{K}_i$ are the new drift and control vector fields of system under the action of the controller ($\alpha, \beta$). The above definition of {\em controlled invariance} is a simple extension of the invariance condition ~(\ref{InvariantDistribution}) for the open loop case. It is now possible to express the necessary and sufficient conditions for the controlled system $(\tilde{K}_0, \tilde{K}_1, \cdots , \tilde{K}_r)$ to be decoupled from the interaction vector field $K_I$ just as we were able to provide conditions for open loop decouplability. The following theorem provides the necessary and sufficient conditions.
\begin{thm}\label{theorem2}
The scalar map $y(t,\xi)=\langle \xi | C(t) |\xi\rangle$ can be decoupled from interaction vector field $K_I$ via suitable analytic control parameters $(\alpha(\xi),\beta(\xi))$ if and only if there exists an involutive distribution $\Delta$ defined on the analytic manifold, such that,
\begin{align}
[K_0, \Delta] \subset \Delta + G \label{cid3}\\
[K_i, \Delta] \subset \Delta + G \label{cid4}
\end{align}
and $K_I \in \Delta \subset \mbox{ker}(dy)$ and where $G=\mbox{ span } \{K_1, \cdots , K_r \}$
\end{thm}
The proof of the above theorem invokes a construction for the sufficiency\cite{ganesan2}, which also provides the means to synthesize the control parameters $\alpha(\xi)$ and $\beta(\xi)$. It is based on this construct that we determine the invariant subspace for the given system and utilize the same to synthesize the parameters. We shall now examine the decouplability of 1-qubit and 2-qubit systems via the conditions stated above and also present an application of the theorem in designing the control for the 2-qubit system in section~(\ref{appendix}).

\section{Examples}
We are interested in studying the decoherence of 1-qubit and 2-qubit systems in the presence of arbitrary controls. Although the theory developed here is general enough to be applicable for all scalar maps represented as bilinear forms, we focus our attention to the decouplability properties of the scalar map which represents the coherence between a set of representative basis states. For the 1-qubit system it is given by $C=|1\rangle\langle 0|$ and for the 2-qubit system it is $C=|10\rangle \langle 01|$. The 1-qubit and 2-qubit system coupled to the environment are modeled as Spin-Boson systems. Consider the 1-qubit system,
\begin{align*}
\frac{\partial \xi(t)}{\partial t} &= \frac{\omega_0}{2}\sigma_z \xi(t) + \sum_k \omega_k b_k^\dagger b_k \xi(t) + u_1 \sigma_x \xi(t) + u_2 \sigma_y \xi(t) + \sum_k \sigma_z(g_k b_k^\dagger + g_k^* b_k)\xi(t)
\end{align*}

with the scalar map, $y(t)=\langle \xi(t) |C| \xi(t)\rangle$, where $C=|1\rangle\langle 0|$, the coherence between the states $|0\rangle$ and $|1\rangle$. The open loop controls, $u_1, u_2$ acting via the Hamiltonians $\sigma_x$, $\sigma_y$ are assumed piecewise constant. We now examine the necessary and sufficient conditions for decouplability via the two approaches developed so far: {\bf Opearator Alegbra:} It can be seen that for the 1-qubit open quantum system, $[C, H_{SE}]=0$ is not satisfied for the given, $C=|1\rangle\langle 0|$ and $H_{SE}$. This is vital for both open loop invariance under arbitrary control($[\tilde{\mathcal{C}}, H_{SE}]=0$, eq.~(\ref{ic})) as well as decouplability via an active controller (lemma~\ref{opralg_controller}).{\bf Invariant Subspace:} The necessary condition for open loop invariance($K_I\in \mbox{ker}(dy)$, threorem~(\ref{invsp}), cond. (ii)) is again not satisfied by the 1-qubit system, as $K_I = \sum_k \sigma_z(g_k b_k^\dagger + g_k^* b_k)\xi(t) \notin \mbox{ker}(dy(t))$ because $L_{K_I}y(t) \neq 0$. This is also vital for decouplability via an active controller~(theorem~(\ref{theorem2})). Hence the conclusion that an open 1-qubit system is not decouplable is independently arrived at by both the formalisms.

Now, consider the following open 2-qubit system,
\begin{align*}
\frac{\partial |\xi(t)\rangle}{\partial t} =& \left( \sum_{j=1}^2 \frac{\omega_0}{2}\sigma_z^{(j)} + \sum_k \omega_k b_k^\dagger
b_k\right)|\xi(t)\rangle + \sum_k \left(\sum_j \sigma_z^{(j)}\right)(g_k b_k^\dagger + g_k^* b_k) |\xi(t)\rangle \\
&+ (u_1(t) \sigma_x^{(1)} + u_2(t) \sigma_y^{(1)} + u_3(t) \sigma_x^{(2)} + u_4(t) \sigma_y^{(2)})|\xi(t)\rangle
\end{align*}
with the scalar map, $y(t)=\langle \xi(t) |C| \xi(t)\rangle$ where $C=|01\rangle\langle 10|$ and the given $H_{SE}$. {\bf {Operator Algebra:}} As was previously shown in section~(\ref{2levelsystems}), this system has a DFS of dimension 2, $\mbox{span}\{|01\rangle, |10\rangle\}$, the states within which remain coherent in the absence of controls. It is also seen that $[C, H_{SE}]=0$. While this part of the necessary condition is satisfied, the sufficient condition for open loop invariance under arbitrary control($[\tilde{\mathcal{C}}, H_{SE}]=0$, eq.~(\ref{ic})) as well as decouplability via an active controller (lemma~\ref{opralg_controller}) is problematic as outlined in section~(\ref{2levelsystems}) and the discussion following lemma~(\ref{opralg_controller}). {\bf {Invariant Subspace:}} Again, it can be seen that the interaction vector field $K_I = \sum_{j,k} \sigma_z^{(j)}(g_k b_k^\dagger + g_k^* b_k) \xi(t)$ belongs to $\mbox{ker}(dy(t))$, (because $L_{K_I} y(t) = 0$) as required by theorem~(\ref{invsp}) for open loop invariance. However for necessary and sufficient conditions~(\ref{cid3})-(\ref{cid4}), we compute the Lie bracket of the control vector fields, $K_1,K_2\in G$, with $K_I\in \mbox{ker}(dy)\subset\Delta$ , and note that,
\begin{align}
[K_{1|2},K_I]&=[\sigma_{x|y}^{(1)}|\xi\rangle, \sum_j \sigma_z^{(j)}(g_k b_k^\dagger + g_k^* b_k) |\xi\rangle]\nonumber \\
&= c. \sum_k \sigma_{y|x}^{(1)}(g_k b_k^\dagger + g_k^* b_k)|\xi\rangle,
\label{requirement}
\end{align}
up to a constant $c$ (where $|$ in the subscript is a placeholder). It can be seen that, $[K_{1|2},K_I]$ does not belong to $\mbox{ker}(dy)$(and hence $\Delta$), nor does it belong to the control distribution $G$. Hence $[K_{1|2},K_I]$ does not belong to $\Delta$ as required by open loop invariance, nor $\Delta + G$  required for decouplability via an active controller. Again the observations made with the help of the two formalisms coincide implying that no collection of two level systems are open loop invariant under arbitrary control nor can be decoupled with an active controller in its present form. However, it can be noted from the aforementioned discussions, that by suitably changing the control Hamiltonians $H_i$(or the control vector fields $K_i$), it is possible to satisfy the necessary and sufficient conditions. Although, the coherence operator $C$ and the interaction Hamiltonian $H_{SE}$ are fixed for a given system, we can modify the control vector fields in order to meet the necessary and sufficient conditions. It is for this reason that the 1-qubit system is not decouplable(as $[C, H_{SE}]\neq 0$) but the 2-qubit system can be decoupled via suitable modifications. The rest of the paper is devoted to studying the means by which such a modification can lead to complete decoherence control. For this purpose, we employ a scalable construction and the quantum controller, the single ancillary qubit in order to affect such a transformation. Henceforth, we will confine ourselves to the study of decouplability of the 2-qubit system and present the ensuing quantum internal model principle.

\section{An Ancillary Quantum Controller}
Consider the following construction employing an ancillary qubit as a quantum controller, with the additional property that its decoherence or strength of the environmental interaction can be modulated externally at will. With this construction it is necessary to maintain only one qubit within a modulated environment, as opposed to a whole system of finite number of qubits. This greatly simplifies the realizability and operability of a practical quantum computer in an ambient setting, without the need for large supercooled environments. This system is now allowed to interact with our qubits of interest through an Ising type coupling $J_1, J_2$ (figure~(\ref{auxsys_schema})). Such a hardware is currently under investigation and development~\cite{nec-riken} with encouraging experimental results. The state vector is now the total wave function of system+ancillary+environment. Both the qubit systems are assumed to interact with the common environment with the only additional requirement that the ancillary qubit's decoherence rate be controllable. Physically this amounts to a coherent qubit with controllable environmental interaction. The scalability and advantages of this construction are analyzed in the next section.
\begin{figure*}[!htpb]
\centering
\subfigure{\includegraphics[width=3.0in, height=1.75in]{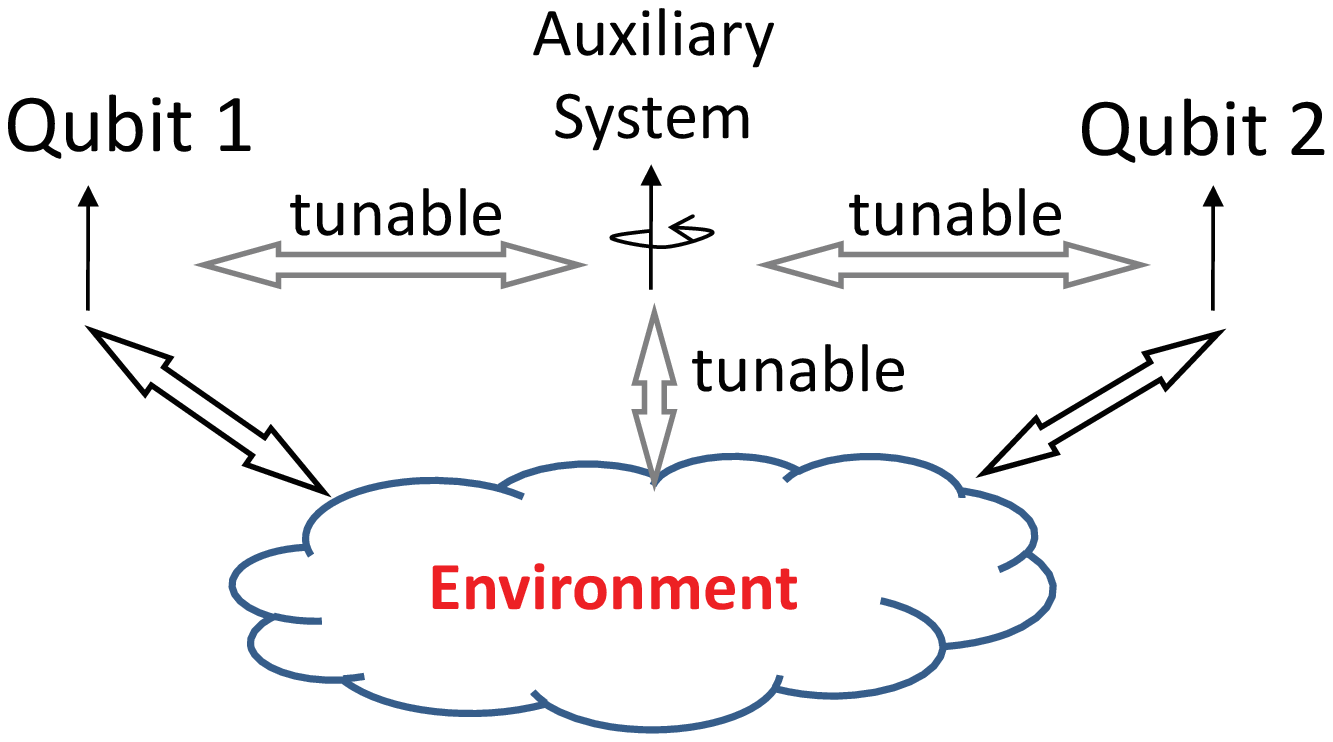}}
\subfigure{\includegraphics[width=3.0in, height=2.5in]{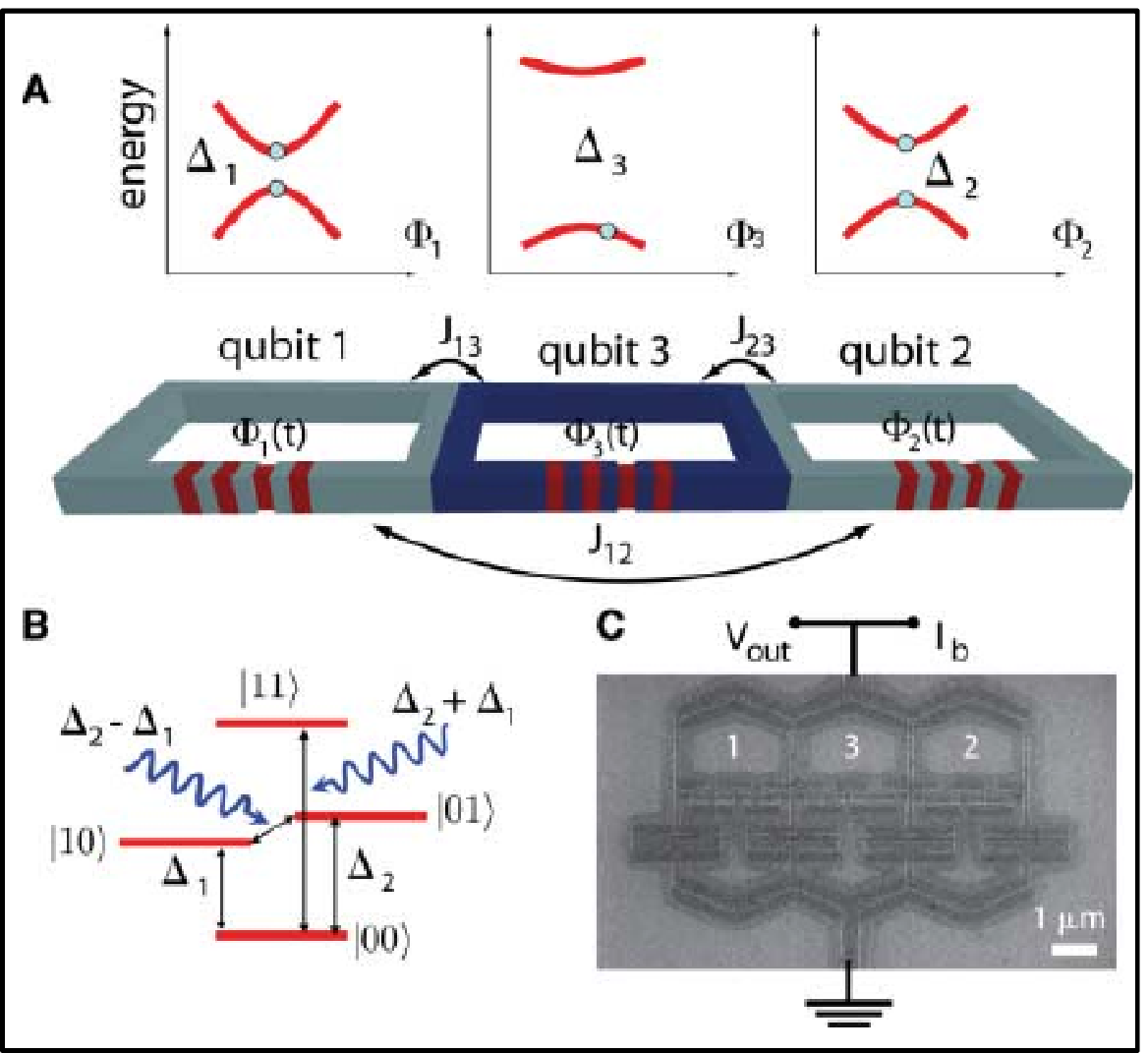}}
\caption{(Left) Schematic of the ancillary quantum controller, that is allowed to interact with the 2-qubit system and the thermal bath via a tunable interaction. (Right) Prospective candidate for experimental implementation of the ancillary quantum controller\cite{nec-riken}. {\scriptsize Courtesy of Y. Nakamura, RIKEN, NEC, Japan.}}\label{auxsys_schema}
\end{figure*}
The control system governing the mechanics following the Schr\"{o}dinger equation~(\ref{auxsys}) is given by,
\begin{align}
&\frac{\partial |\xi(t)\rangle}{\partial t} =\left( \sum_{j=1}^2 \frac{\omega_0}{2}\sigma_z^{(j)} + \sum_k
\omega_k b_k^\dagger b_k\right)\xi(t) + \sum_{j,k}\sigma_z^{(j)}(g_k b_k^\dagger + g_k^* b_k)\xi(t)+ \left(u_1(t) \sigma_x^{(1)} + u_2(t) \sigma_y^{(1)}\right. \nonumber \\
&+ u_3(t) \sigma_x^{(2)}+ u_4(t) \sigma_y^{(2)}+ \frac{\omega_0}{2}\sigma_z^{(b)} + u_5\sigma_x^{(b)} \left.+u_6\sigma_y^{(b)}+u_7 J_1 \sigma_z^{(1)}\sigma_z^{(b)}+ u_8 J_2 \sigma_z^{(2)}\sigma_z^{(b)}\right)\xi(t)\nonumber \\
&+ u_9 \sum_k \sigma_z^{(b)}(w_k b_k^\dagger + w_k^* b_k)\xi(t) \label{auxsys}
\end{align}
with $\sigma_{x|y|z}$now skew Hermitian and the same scalar map as before. The superscripts $^{(1)}, ^{(2)}$ and $^{(b)}$ denote operators acting on the Hilbert spaces of the first qubit, second qubit and the ancillary qubit respectively. It is now seen that,
\[
[K_{1|2},K_I]= c. \sum_k \sigma_{y|x}^{(1)}(g_k b_k^\dagger + g_k^* b_k) |\xi\rangle \mbox{ (from eq.~(\ref{requirement}))}
\]
belongs to the control algebra generated by the additional vector fields introduced by the ancillary system, (i.e), $[K_i, K_I]$ belongs to the Lie algebra generated by the control vector fields $K_1, \cdots, K_9$ of the above system. Hence restructuring the system such that the linear span of the control vector fields and Lie algebra of the control vector fields coincide would ensure that the necessary and sufficient conditions given by equations~(\ref{cid3}) and (\ref{cid4}) are satisfied.

\section{The Restructured Quantum Control System}
The ancillary qubit is primarily used to generate vector fields that can help decouple the system from the vector field $K_I$. The additional vector fields that can be generated with the help of ancillary system, help satisfy the necessary and sufficient conditions for decouplability as shown in this section.  Let,
\[
H_0 =\sum_{j=1}^2 \frac{\omega_0}{2}\sigma_z^{(j)} + \sum_k \omega_k b_k^\dagger
b_k +\frac{\omega_0}{2}\sigma_z^{(b)}
\]
denote the Hamiltonians of, qubits 1\&2, environment and the ancillary system,
\[
H_{SE} = \sum_k\left(\sum_j \sigma_z^{(j)}\right)(g_k b_k^\dagger + g_k^* b_k)
\]
the System + Environment decoherence Hamiltonian,
\[
H_1 = \sigma_x^{(1)}, H_2 = \sigma_y^{(1)}, H_3 = \sigma_x^{(2)}, H_4 =
\sigma_y^{(2)}
\]
the control Hamiltonians acting on qubits 1 and 2,
\[
H_5 = \sigma_x^{(b)}, H_6 = \sigma_y^{(b)}, H_7 = J_1
\sigma_z^{(1)}\sigma_z^{(b)}, H_8 = J_2
\sigma_z^{(2)}\sigma_z^{(b)}
\]
the control Hamiltonians for the ancillary system, the Ising type coupling to qubits 1 and 2, and
\[
H_9 = \sum_k \sigma_z^{(b)}(w_k b_k^\dagger + w_k^* b_k)
\]
the controllable interaction of the ancillary qubit with the environment.

These 9-control Hamiltonians($H_1\cdots H_9$) along with $H_0$ and $H_{SE}$ govern the evolution of the system. The controls are implemented by the actual hardware(ancillary quantum controller) and corresponding fields $(u_1,\cdots u_9)$ with $u_7$ and $u_8$ being the strength of the Ising coupling. With this setup it is now possible to {\em generate} additional control vector fields by suitably manipulating the field strengths. With the additional controls generated, it is possible to satisfy the necessary and sufficient conditions~(\ref{cid3}, \ref{cid4}). This way we are able to come up with the "restructured" quantum control system via the following "control pulse" maneuvers. For example, consider the following maneuver, with $u_6$ and $u_9$,
\begin{align*}
&u_6(\tau)=1, \mbox{ and } u_9(\tau)=0, \mbox{ for } \tau\in[0,t]\\
&u_6(\tau)=0, \mbox{ and } u_9(\tau)=1, \mbox{ for } \tau\in[t,2t]\\
&u_6(\tau)=-1, \mbox{ and } u_9(\tau)=0, \mbox{ for } \tau\in[2t,3t]\\
&u_6(\tau)=0, \mbox{ and } u_9(\tau)=-1, \mbox{ for } \tau\in[3t,4t]
\end{align*}
The corresponding unitary time evolution operator at the end of time instant $4t$ is given by,
\begin{align*}
U(4t) &= e^{(-iH_6 t)}e^{(-iH_9 t)}e^{(iH_6 t)}e^{(iH_9 t)}\\
&= \exp(-i[H_6,H_9]t^2+\mathcal{O}(t^3))
\end{align*}
the series expansion by Campbell-Baker-Hausdorff formula. In the limit that $t=dt\rightarrow0$. The effective direction of evolution is given by the commutator of the corresponding Hamiltonians, but to the second order in time. Hence we can devise a control vector field in the direction given by the commutators of the corresponding Hamiltonians $H_6$ and $H_9$, where,
\[
[H_6,H_9] = c.\sigma_x^{(b)}\sum_k(w_k b_k^\dagger + w_k^* b_k)
\]
where $c$ is a {\it real} constant for a skew Hermitian $H_6$ and $H_9$. In fact it is possible to generate a direction of evolution with arbitrary strength corresponding to repeated commutators of the Hamiltonians $H_1\cdots H_9$ of the
physical system~(\ref{auxsys}). The commutators of tensor product of operators are calculated according to,
\begin{align*}
[A\otimes B, C\otimes D] = CA\otimes [B,D]+[A,C]\otimes BD
\end{align*}
With the control field $H_8$ we can generate the following direction in conjunction with the previous maneuver
$[H_8,H_5]=c'J_2\sigma_z^{(2)}\sigma_y^{(b)}$ and also,
\begin{align}
[[H_8,H_5],[H_6,H_9]]&=
c_1.[J_2\sigma_z^{(2)}\sigma_y^{(b)},\sigma_x^{(b)}\sum_k(w_k b_k^\dagger +
w_k^* b_k)]\nonumber\\
&=c.\sigma_z^{(2)}\sigma_z^{(b)}\sum_k(w_k b_k^\dagger + w_k^* b_k)
\label{comm1}
\end{align}
A similar maneuver between controls $u_4,u_6 \mbox{ and }u_8$, generates the following direction of evolution,
\begin{align}
[H_4,H_8] &= [\sigma_y^{(2)},J_2\sigma_z^{(2)}\sigma_z^{(b)}] =
c.\sigma_x^{(2)}\sigma_z^{(b)} \label{comm2}
\end{align}
where $c$ is a real constant for a skew Hermitian $H_4, H_8$. Again, from operating on equations~(\ref{comm1}) and (\ref{comm2}) we get,
\begin{align}
[[H_4,H_8],[[H_8,H_5],[H_6,H_9]]] = &c_1[\sigma_x^{(2)}\sigma_z^{(b)},
\sigma_z^{(2)}\sigma_z^{(b)}\sum_k(w_k b_k^\dagger + w_k^* b_k)]\nonumber\\
&=c_1.[\sigma_x^{(2)},\sigma_z^{(2)}].(\sigma_z^{(b)})^2.\sum_k(w_k b_k^\dagger
+ w_k^* b_k)\nonumber\\
&=c.\sigma_y^{(2)}.\mathbb{I}^{(b)}.\sum_k(w_k b_k^\dagger + w_k^*
b_k)\label{comm3}
\end{align}
where $\mathbb{I}^{(b)}$ is the identity operator on the ancillary subsystem. Hence we have generated an effective controllable coupling between $\sigma_y^{(2)}$ and the environment with the help of the ancillary qubit. It is important to note that the Hamiltonian so obtained by the above control maneuver now acts trivially on the Hilbert space of the ancillary qubit, a property which is found to be extremely useful. It is also possible to generate the $\sigma_x^{(2)}$ counterpart of the above coupling by a similar maneuver, given by,
\begin{align}
c.\sigma_x^{(2)}.\mathbb{I}^{(b)}.\sum_k(w_k b_k^\dagger + w_k^*
b_k)\label{comm4}
\end{align}
Again by a symmetric and identical argument we can generate a coupling between the environment and qubit 1, which is given by,
\begin{equation}
c.\sigma_y^{(1)}.\mathbb{I}^{(b)}.\sum_k(w_k b_k^\dagger + w_k^* b_k) \mbox{ and }
c.\sigma_x^{(1)}.\mathbb{I}^{(b)}.\sum_k(w_k b_k^\dagger + w_k^* b_k)\label{comm6}
\end{equation}
It can be seen that the above vector fields are what are required in equation~(\ref{requirement}) for the Lie bracket $[K_{1|2},K_I]$ to be contained within $\Delta + G$. Now, noting that the constants $c$ in the above equations can be controlled independently and arbitrarily, we can write the preliminary form of the {\it actual} control system which achieves disturbance decoupling. Gathering terms (\ref{comm3})-(\ref{comm6}), we construct the following control system for $\frac{\partial |\xi(t)\rangle}{\partial t}$ given by equation~(\ref{restructured_system}). In the following control system, the environment is approximated to be of single mode and of three energy levels\cite{ganesan2}. The vast majority of the interaction energy is stored in the fundamental mode and first few energy states of the oscillator.
\begin{figure*}[!htp]
\hrulefill
\normalsize
\begin{align}
\frac{\partial |\xi(t)\rangle}{\partial t} =& \left( \sum_{j=1}^2 \frac{\omega_0}{2}\sigma_z^{(j)} + \sum_k \omega_k b_k^\dagger
b_k\right)|\xi(t)\rangle + \sum_{j=1}^2 \sigma_z^{(j)}(g b^\dagger + g^* b) |\xi(t)\rangle \nonumber\\
&+ \sum_{i=0}^2 u_{1i}\sigma_x^{(1)}(w b^\dagger + w^* b)^i|\xi(t)\rangle + \sum_{i=0}^2 u_{2i}\sigma_y^{(1)}(w b^\dagger + w^* b)^i|\xi(t)\rangle \nonumber\\&
+ \sum_{i=0}^2 u_{3i}\sigma_x^{(2)}(w b^\dagger + w^* b)^i|\xi(t)\rangle + \sum_{i=0}^2 u_{4i}\sigma_y^{(2)}(w b^\dagger + w^*
b)^i|\xi(t)\rangle\nonumber\\&+ \sum_{i=0}^2 u_{5i}\sigma_x^{(1)}\sigma_z^{(2)}(w b^\dagger + w^* b)^i|\xi(t)\rangle + \sum_{i=0}^2 u_{6i}\sigma_y^{(1)}\sigma_z^{(2)}(w b^\dagger + w^* b)^i|\xi(t)\rangle\nonumber\\&
+ \sum_{i=0}^2 u_{7i}\sigma_z^{(1)}\sigma_x^{(2)}(w b^\dagger + w^* b)^i|\xi(t)\rangle + \sum_{i=0}^2 u_{8i}\sigma_z^{(1)}\sigma_y^{(2)}(w b^\dagger + w^* b)^i|\xi(t)\rangle \label{restructured_system}
\end{align}
\hrulefill
\vspace*{4pt}
\end{figure*}
The 24 restructured controls $u_{10}\cdots u_{12},u_{20}\cdots u_{22},\cdots,u_{30},\cdots u_{82}$ could be thought of as "software" generated controls by manipulation of strengths of actual fields from system~(\ref{auxsys}). Though, in the system above, we consider only first 3 states of the environment, it is possible to consider any finite number of environmental interaction terms by simply including additional control terms in system~(\ref{restructured_system}), via the same analysis. We are now in a position to use the new controls to decouple the scalar map from the environmental interaction. This restructured system satisfies the necessary and sufficient condition for decouplability because the disturbance vector field, $K_I = H_{SE}|\xi\rangle$, is contained within $\mbox{ker}(dy)$ where $y=\langle \xi|01\rangle \langle 10|\xi\rangle$. To see this, one can evaluate $L_{K_I} y(t)$ and notice that it vanishes. The sufficient condition can be seen from the fact that $[K_I, K_i] \in G$, where $G$ is span of control vector fields of system~(\ref{restructured_system}). Hence the necessary and sufficient conditions,
\begin{eqnarray}
&&(i) K_I\in \Delta \subset \mbox{ker}(dy)\nonumber\\
&&(ii) [K_I, K_i]\in \Delta + G, \mbox{ where } G=\mbox{span}\{K_1\cdots,K_{24}\}
\end{eqnarray}
are satisfied. \\
\textbf{Scalability:} It can seen that in order to perform the restructuring of a quantum system with finite number of qubits, it is necessary to employ only one ancillary quantum controller. This guarantees the scalability of the construction and decouplability properties of resulting system.

In summary, the finite system and environment approximation has enabled us to come up with the control system whose coherence can be perfectly decoupled from the environmental interaction as shown in section~(\ref{results}). By ensuring the maximum rank of control matrix $\beta$, the controllability properties of the original 2-qubit quantum system on the manifold is fully preserved. In summary, given the existence of an invariant subspace $\Delta \subset \mbox{ker}(dy)$, the above conditions for decouplability can be summarized as,
\begin{center}
\begin{tabular}{|c|c|c|}
\hline
Open loop Uncontrolled &$[\Delta,K_0]\subset\Delta$ &$[\Delta,K_i]\subset\Delta$\\\hline
Controlled &$[\Delta,K_0]\subset\Delta+G$ &$[\Delta,K_i]\subset\Delta+G$\\\hline 
\end{tabular}\\
\end{center}
The analyses performed thus far also allows us to summarize the results in the form of the following corollary.
\begin{cor}
For any finite $N$-qubit open quantum system acted upon by arbitrary user generated control via Pauli matrices and under the influence of decoherence interaction $H_{SE}$, the coherence between the basis states $|i\rangle$ and $|j\rangle$ cannot be rendered immune to $H_{SE}$ without the action of the quantum controller which is the ancillary quantum system.
\end{cor}

\section{Results}\label{results}
\begin{figure}[!htpb]
\begin{center}
\epsfysize=3in \epsfxsize=3.75in \epsffile{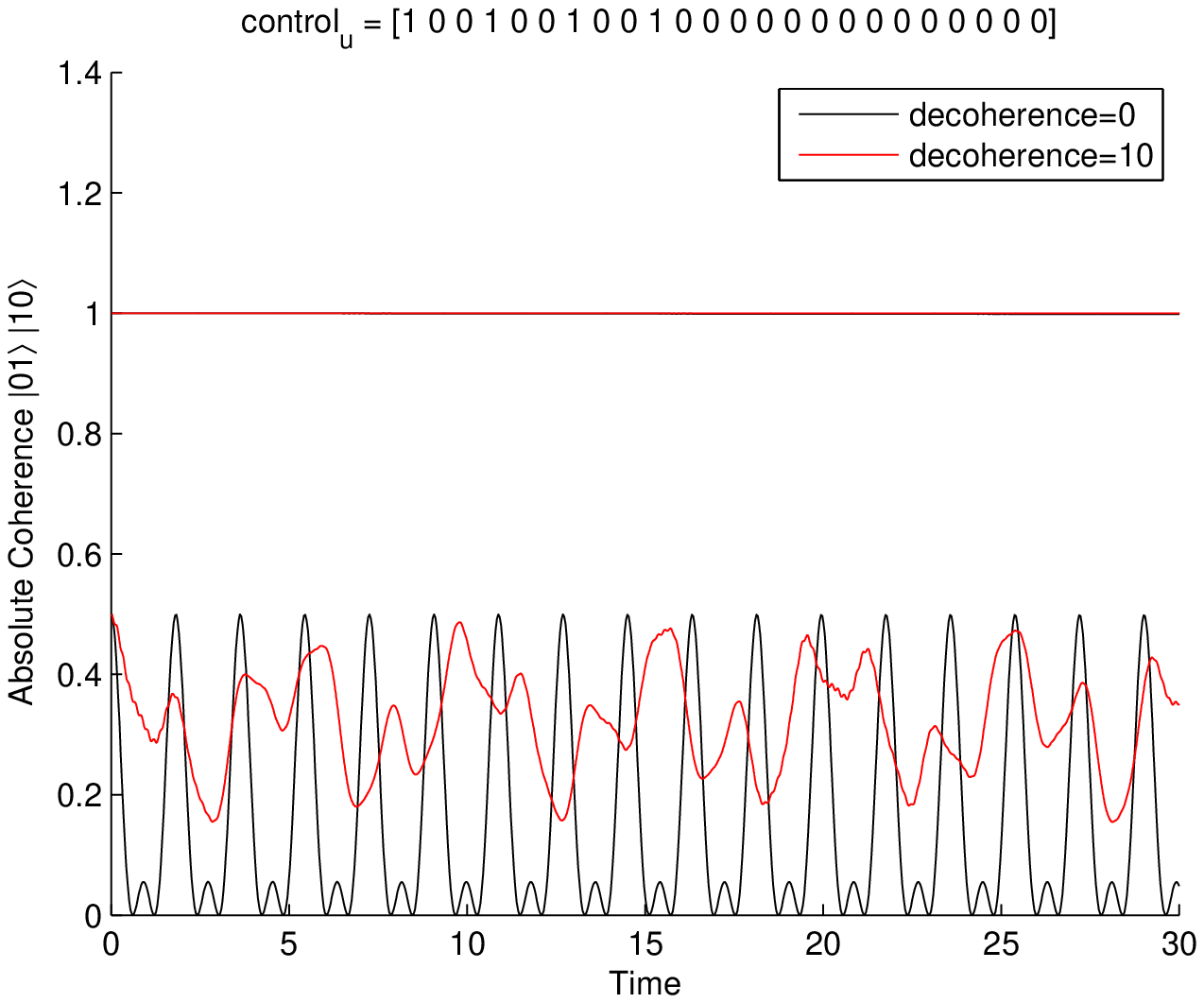}
\end{center}
\caption{Open loop behavior of the 2 qubit system interacting with 3 level
environment}\label{openloopbehavior}
\end{figure}
\begin{figure*}[!htpb]
\centering
\subfigure{\includegraphics[width=3.20in, height=3in]{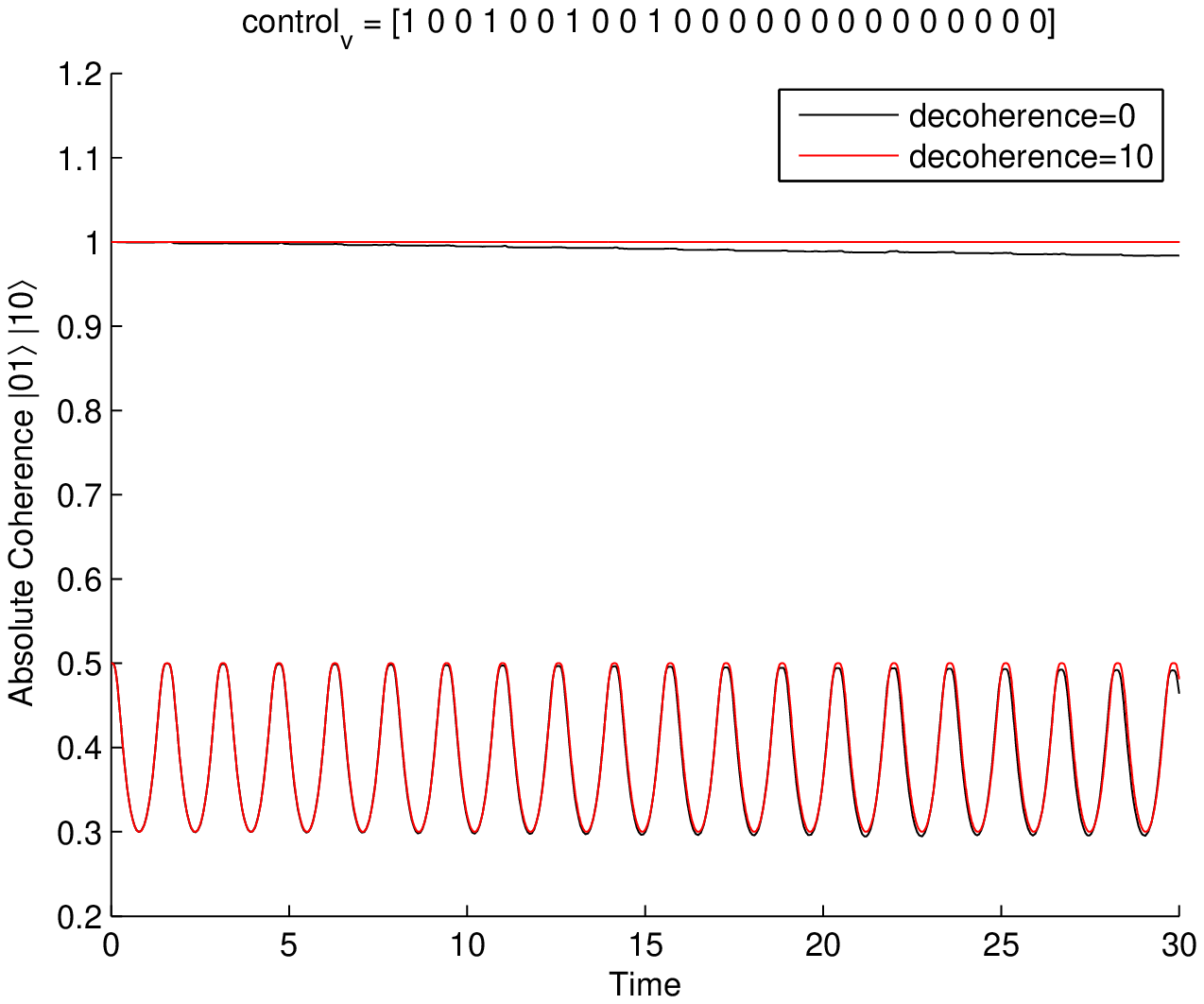}}
\subfigure{\includegraphics[width=3.20in, height=3in]{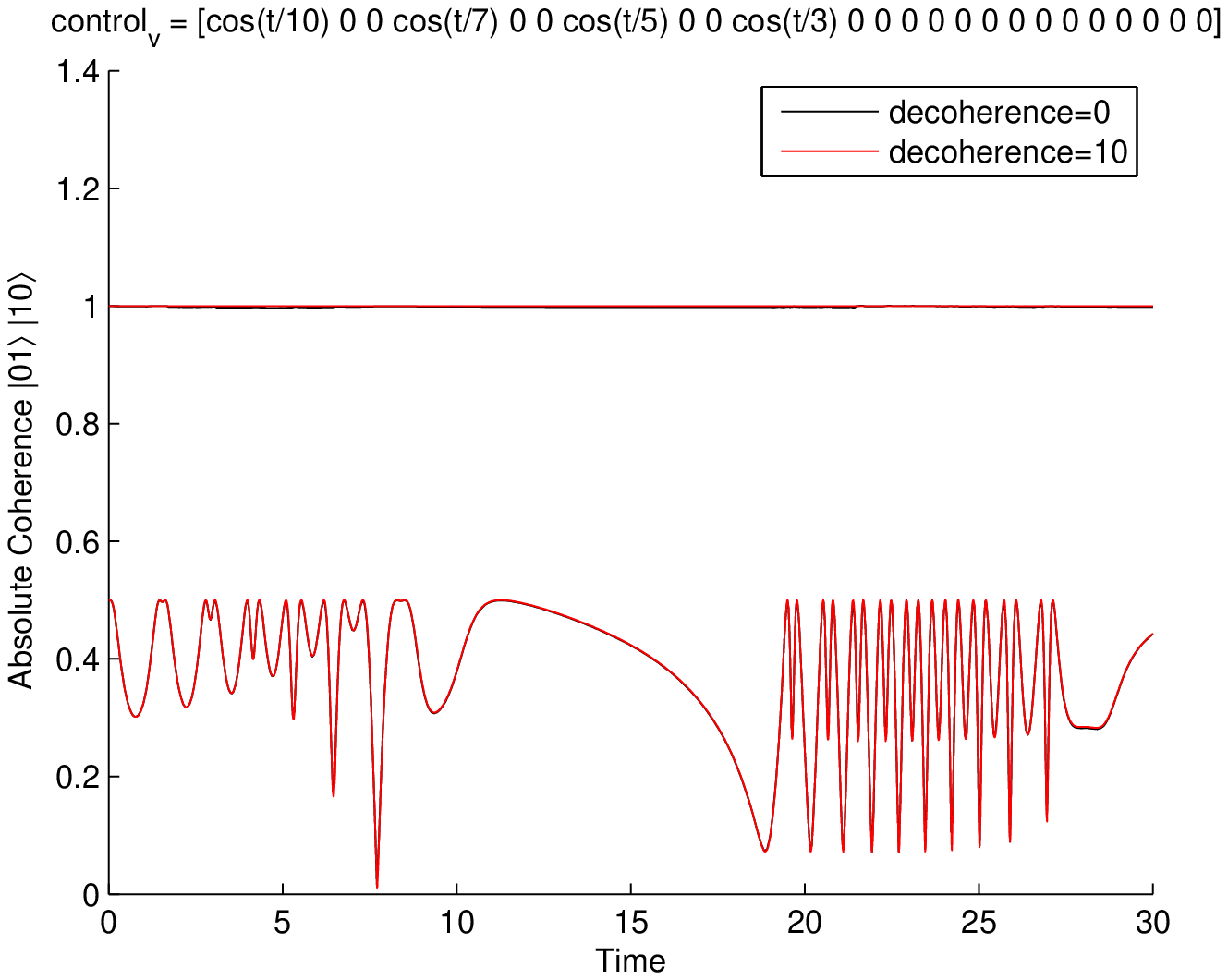}}
\subfigure{\includegraphics[width=3.20in, height=3in]{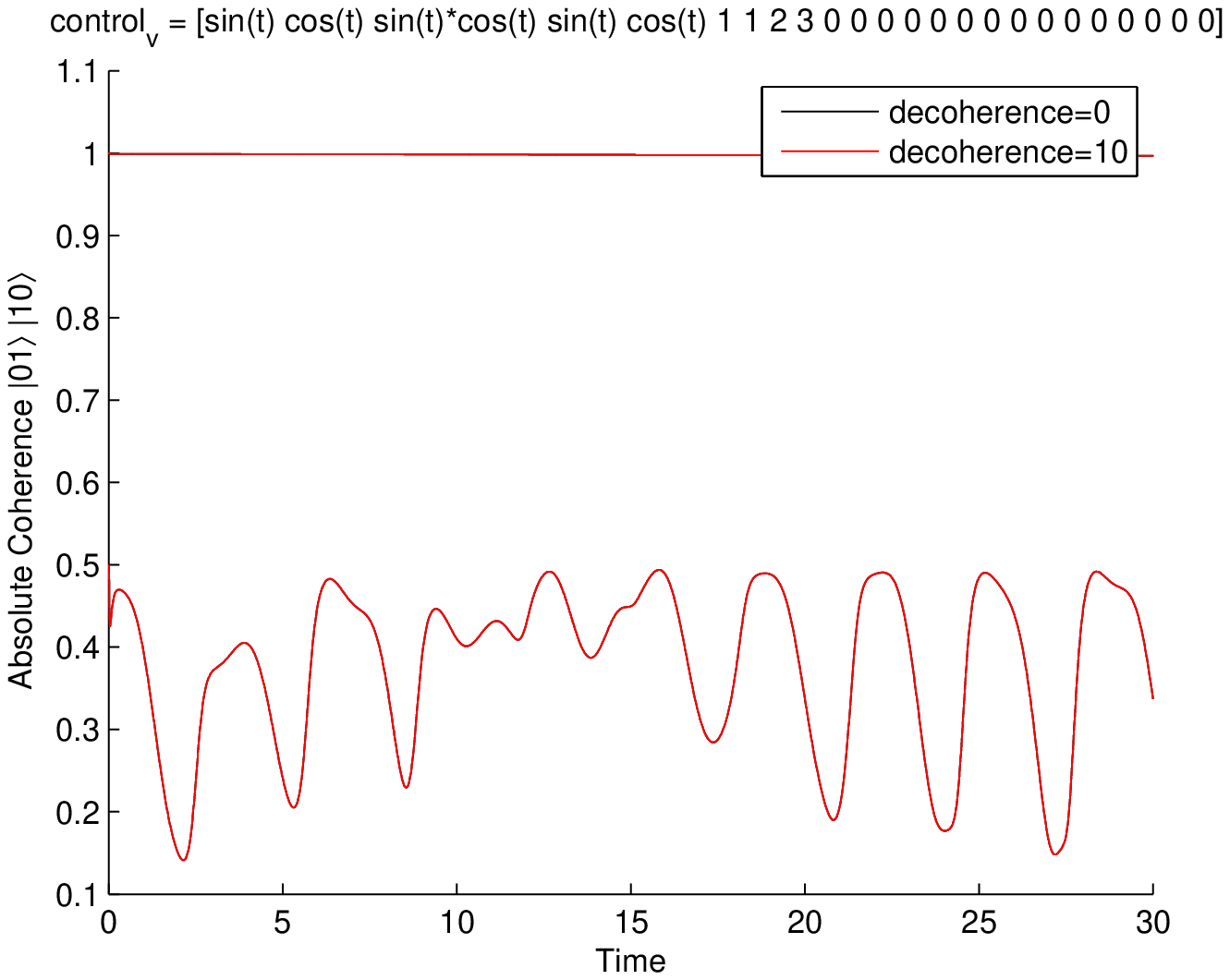}}
\subfigure{\includegraphics[width=3.20in, height=3in]{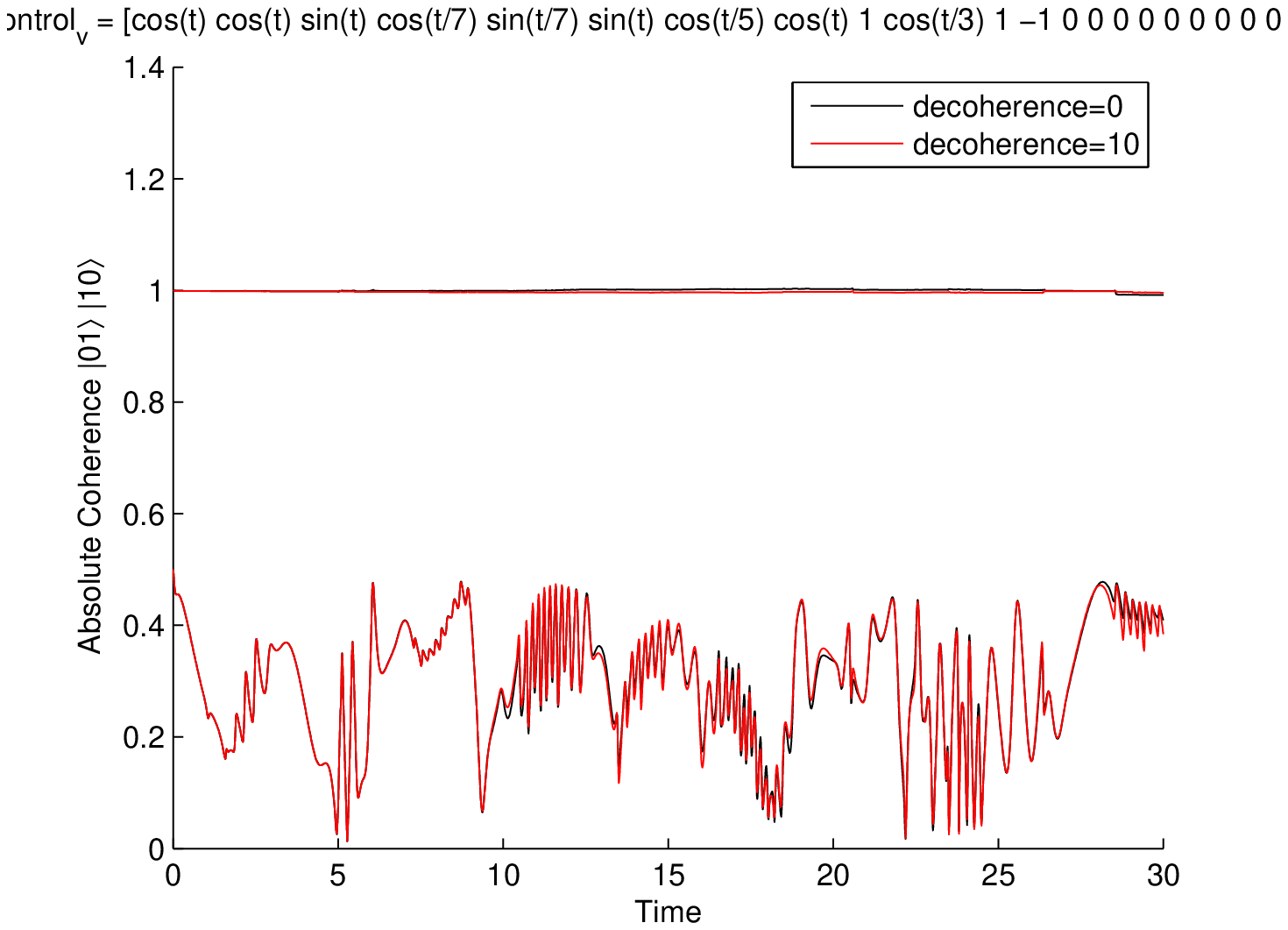}}
\subfigure{\includegraphics[width=3.20in, height=3in]{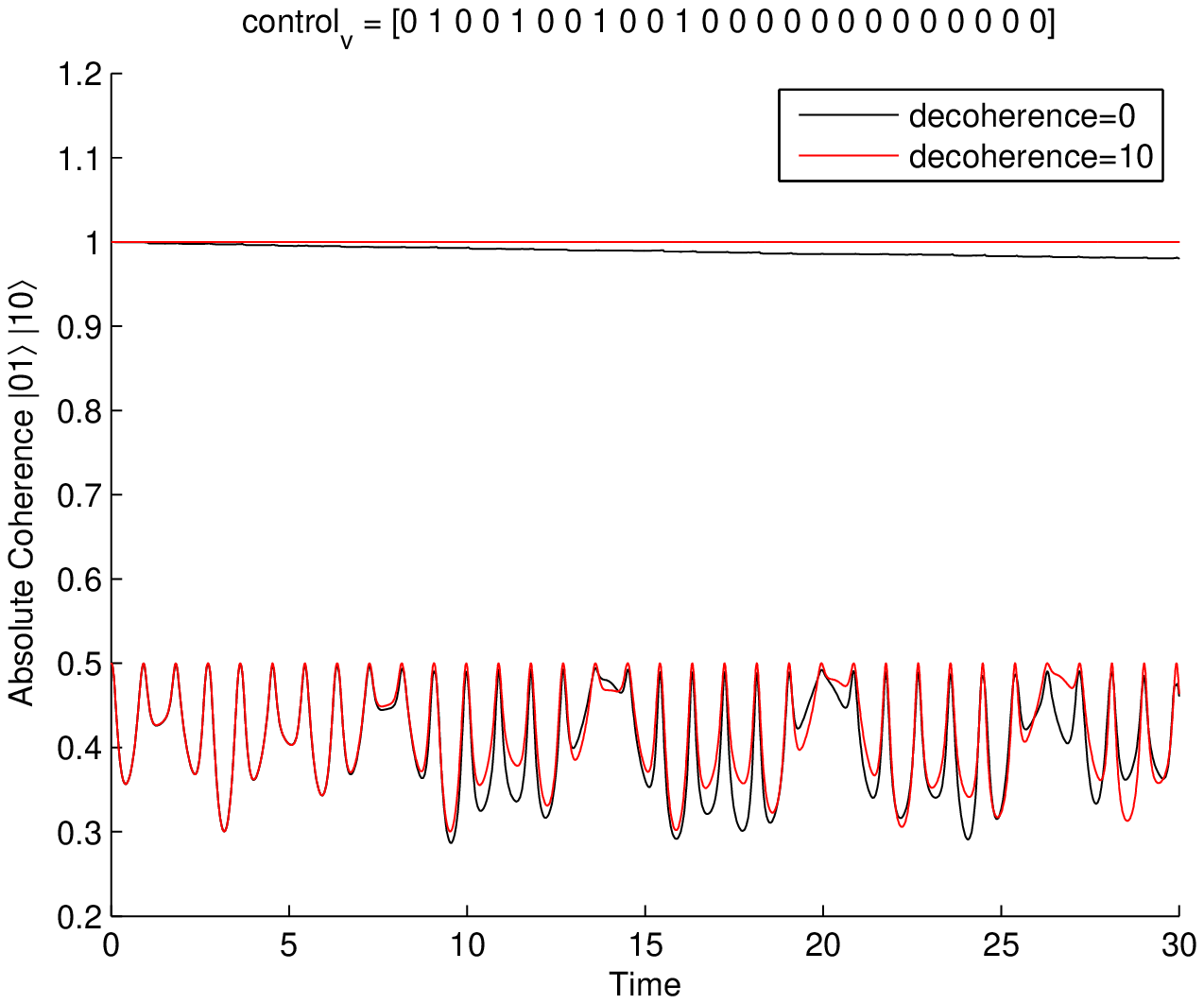}}
\subfigure{\includegraphics[width=3.20in, height=3in]{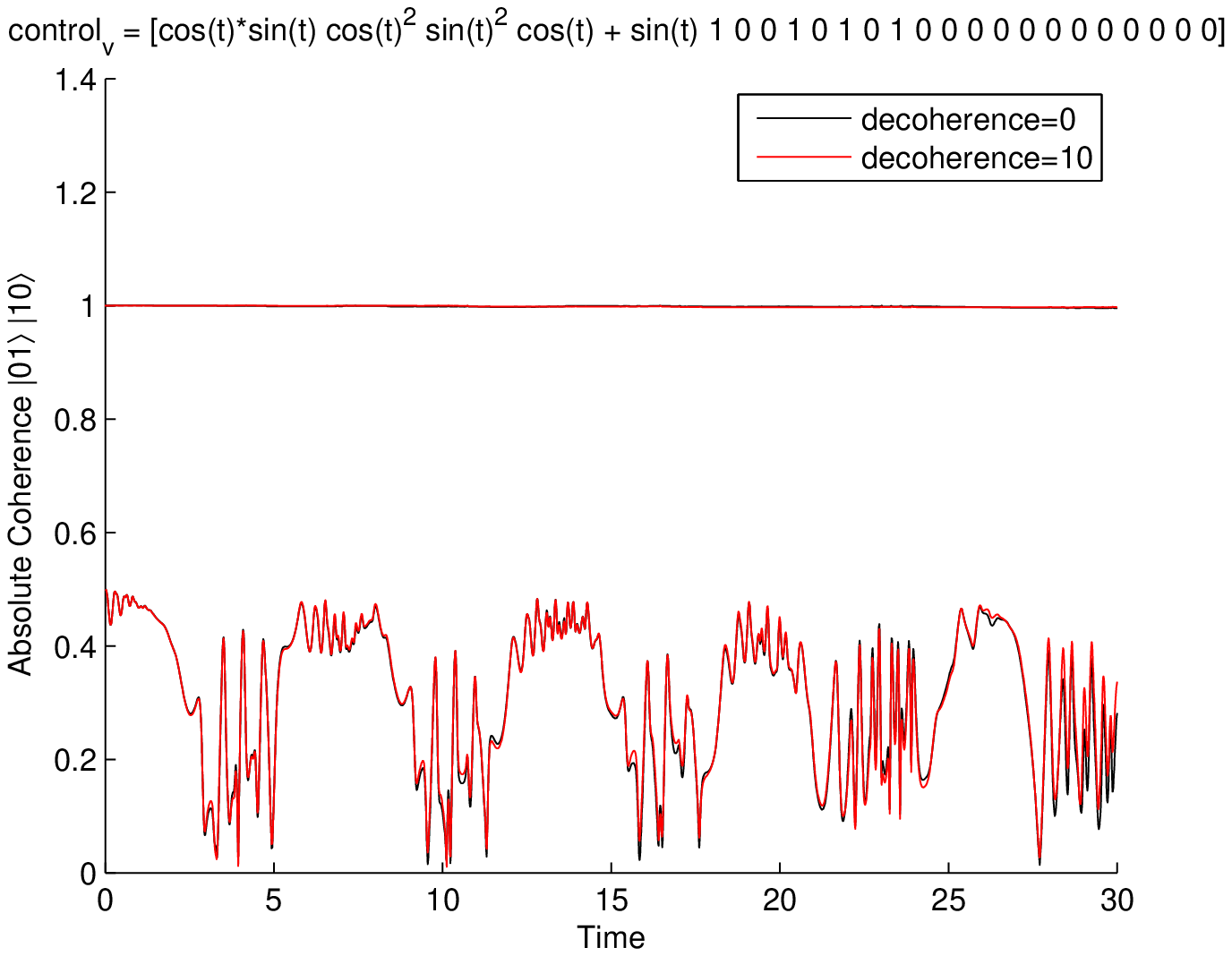}}
\caption{Coherence between the basis states $|01\rangle$ and $|10\rangle$ of a 2-qubit system under the action of the ancillary quantum controller is preserved under arbitrary controls. The controls were chosen to be either constant, time varying or both.}
\label{controlled_behavior}
\end{figure*}
The above system was simulated with 2, two-level interacting systems of interest and 1 ancillary qubit interacting with the environment. The goal was to study the effect of decoherence Hamiltonian on the coherence between $|01\rangle$ and $|10\rangle$ under {\em arbitrary} control as stated in the problem statement. We present the simulation results based on the above control strategy for two different strengths of decoherence interaction (red, for interaction strength 10 and black for no interaction) and different(both constant and time varying) control strengths. We approximate the environment by a single mode and first three levels for simplicity, as the vast majority of the environmental energy, modeled as infinite harmonic oscillators in different modes, is stored within the lowest mode and in lower energy states. This can be seen by the nature of the decoherence  interaction\cite{openqusys} and the coherent state of the harmonic oscillator\cite{louisell}. However, it can also be seen that, in order to include the higher order environmental interactions, one merely needs to add more control terms, following the same analysis. In this work, we consider the particular model of environmental interaction only to demonstrate the applicability of the technique and clarity of presentation. The initial coherence of the state between $|10\rangle$ and $|01\rangle$ is set to 0.5. The absolute value of the coherence with and without the decoherence interaction is presented along with the norm of the state in both the cases. Figure~(\ref{openloopbehavior}) is the open loop behavior of the 2-qubit system and Figure~(\ref{controlled_behavior}) is the behavior of the decouplable system~(\ref{restructured_system}), acted upon by the ancillary quantum controller. For the open loop system the coherence between states $|10\rangle$ and $|01\rangle$ is influenced in the presence of decoherence interaction. Though the coherence is irretrievably lost in reality, the periodicity is due to the finiteness of the system and environment. With the action of ancillary quantum controller and synthesis of control parameters as outlined in section~\ref{appendix}, we see identical behavior of the coherence for different values of decoherence strength and useful control, thus effectively rendering the coherence between states $|10\rangle$ and $|01\rangle$ immune to $H_{SE}$. The behavior is seen to match for any set of analytical control functions, thus achieving perfect decoupling. In addition, the value of $\hbar$ in the Schr\"{o}dinger equation was set to '1', which gives rise to scaled time in the simulation. The slight deviation of the norm from '1' is an artifact of the numerical imperfections.

\section{Quantum Internal Model Principle}
In this section we present the need for the Quantum Internal Model principle. Classical internal model principle for linear systems~\cite{francis}\cite{huangjie} outlines conditions for robust tracking in the presence of disturbance $d(t)$. Consider the linear system with disturbance, 
\begin{eqnarray}
\dot{x} = Ax+Bu+E_d d,\\
y=Cx+Du+F_d d.
\end{eqnarray}
with the tracking error $e=Cx+Du+F_d d - r$. The {\em exosystem} consists of the reference input and the plant noise, both generated by linear autonomous differential equations,
\begin{eqnarray}
\dot{r} = A_{1r} r, r(0) = r_0;\\
\dot{d} = A_{1d} d, d(0) = d_0
\end{eqnarray}
with arbitrary initial states. The robust output regulation, where the tracking error is driven to zero, not only requires a dynamic state feedback but also that the controller {\em mimic} the exosystem in terms of its characteristic polynomial. Classical disturbance decoupling~\cite{isidori}\cite{isidori1}\cite{isidori2}, on the other hand requires only the knowledge of system parameters and not the model of disturbance in order to completely decouple the output. However, quantum decoherence control, which is similar in formulation to classical disturbance decoupling is possible only with the knowledge of the environmental interaction, which is analogous to classical robust output regulation. This allows us to propose the quantum internal model principle with the following characteristics,\\
{\bf Quantum and Classical Internal Model Principle}
\begin{itemize}
\item Quantum Internal Model principle aims at {\em disturbance rejection} with the help of a ancillary quantum controller and the knowledge of the model of interaction with the environment.
\item Classical Internal Model principle aims at perfect trajectory tracking via feedback, which involves the knowledge of the disturbance generator (as well as the desired trajectory generator) viz. the exosystem,
\end{itemize}
{\bf Quantum and Classical Disturbance Rejection}
\begin{itemize}
\item Quantum disturbance decoupling, which is the underlying motivation of Quantum Internal Model requires complete knowledge of the model of the environment as well the corresponding model of decoherence, in the combined system+environment state space.
\item Classical disturbance decoupling only requires the model of interaction within the system's state space.
\end{itemize}
Hence the knowledge of environmental interaction with the system for complete decoupling makes quantum internal model principle salient and important within the framework of systems and control. The original 2-qubit system had to be augmented with the (internal) model of the environment entering via the control $u_9$ in equation~(\ref{auxsys}) so as to restructure the vector fields to act non-trivially on the environment Hilbert space. Hence the knowledge of the model of interaction with the environment, i.e, the decoherence Hamiltonian $H_{SE}$ is essential to successfully controlling decoherence.
\begin{figure}[!htpb]
\begin{center}
\epsfysize=1in \epsfxsize=2in \epsffile{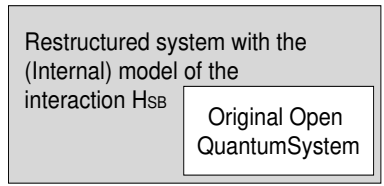}
\end{center}
\caption{The original open quantum system acts as the skeletal structure for the larger restructured system.}\label{restructure}
\end{figure}
Figure~(\ref{restructure}) depicts the nature of the original and restructured systems in the sense that the latter is larger and is derived from the original by taking into account the model of the environmental interaction $H_{SE}$, whose closest classical analog is the disturbance generator denoted by $A_{1d}$. The Figure~(\ref{modelfig}) outlines the schematic of control system for the decoupling problem, where the coherence measure for the controlled open quantum system and the corresponding closed system are identical. In summary, the structure of the system needed to be altered in order to,
\begin{itemize}
\item Artificially induce coupling between qubits $1$, $2$ and the environment with the help of the ancillary qubit.
\item Generate vector fields in the higher order of the environment operator via a fast-action open-loop control.
\end{itemize}
Hence it was necessary to modify the core system in more ways than one in order to perform decoupling. It is to be noted the above control strategy is a hybrid of fast-action open loop control and smooth analytic control ($\alpha$ and $\beta$) in order to achieve perfect decoherence elimination.
\begin{figure}[!htpb]
\begin{center}
\epsfysize=1.75in \epsfxsize=3.0in \epsffile{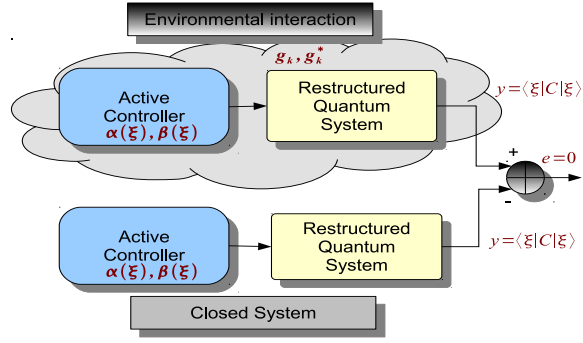}
\end{center}
\caption{The difference between coherence measures from the open quantum system and the closed quantum systems must vanish.}\label{modelfig}
\end{figure}
\section*{Conclusion}
In this article we visited an Internal Model Principle that is uniquely related to quantum systems in light of disturbance decoupling and decoherence control. The tensorial model of interaction of the quantum system with the environment can be skillfully exploited to completely decouple the system from the same. Such a result and its implication are first of its kind in the literature to the best of authors' knowledge. The ideas here presented could not only help further decoherence control but also influence the design of future quantum and classical control systems. In addition a framework for enhanced disturbance decoupling was laid wherein the entire control algebra can be used to effectively decouple a larger class of systems than just the linear span of the control vector fields.


\section{Appendix}\label{appendix}
By following the proof of Theorem~(\ref{theorem2}) above as outlined in \cite{ganesan2}, we synthesize the control parameters  $\alpha(\xi)$ and $\beta(\xi)$ for the System~(\ref{restructured_system}). Let the restructured control vector fields in the system Hilbert Space be given by,
\begin{equation*}
\begin{array}{cc}
g_{1s} = \sigma^{(1)}_x|\xi\rangle; &g_{5s} = \sigma^{(1)}_z\otimes\sigma^{(2)}_x|\xi\rangle\\
g_{2s} = \sigma^{(1)}_y|\xi\rangle; &g_{6s} = \sigma^{(1)}_z\otimes\sigma^{(2)}_y|\xi\rangle\\
g_{3s} = \sigma^{(2)}_x|\xi\rangle; &g_{7s} = \sigma^{(1)}_x\otimes\sigma^{(2)}_z|\xi\rangle\\
g_{4s} = \sigma^{(2)}_y|\xi\rangle; &g_{8s} = \sigma^{(1)}_y\otimes\sigma^{(2)}_z|\xi\rangle
\end{array}
\end{equation*}
along with three environmental operators $(\mathbbm{I},D,D^2)$, we get 24
control vector fields,
\begin{align*}
\{g_1,\cdots g_{24}\} &= \{\sigma^{(1)}_x, \sigma^{(1)}_y, \sigma^{(2)}_x, \sigma^{(2)}_y, \sigma^{(1)}_z\otimes\sigma^{(2)}_x,\\
&\sigma^{(1)}_z\otimes\sigma^{(2)}_y, \sigma^{(1)}_x\otimes\sigma^{(2)}_z, \sigma^{(1)}_y\otimes\sigma^{(2)}_z\}\otimes \{\mathbbm{I},D,D^2\}|\xi\rangle
\end{align*}
where the multiplication is carried out in the usual order, where $D=(wb^\dagger+w^*b)$, is the displacement operator for Quantum Harmonic Oscillator. Define $G = \{g_1,\cdots, g_{24}\}$. The system and environmental identity operators are suppressed for ease of notations and they are assumed to be present where it is clear from the context. The above control vector fields are generated by "software" in that, the control vector fields are produced by maneuvering action of fast action control pulses applied to the ancillary qubit enhanced open loop system~(\ref{auxsys}).

In order to construct the invariant subspace for the restructured quantum control system with the control vector fields as above, it can be seen that the vector fields,
\begin{eqnarray}
\delta_1 &=& (\sigma^{(1)}_z+\sigma^{(2)}_z)|\xi\rangle,\nonumber\\
\delta_2 &=& \sigma^{(1)}_z\otimes\sigma^{(2)}_z|\xi\rangle\nonumber\\
\delta_3 &=& \mathbbm{I}|\xi\rangle,\nonumber\\
\delta_4 &=& (\sigma^{(1)}_x\otimes \sigma^{(2)}_x - \sigma^{(1)}_y\otimes \sigma^{(2)}_y)|\xi\rangle\nonumber\\
\delta_5 &=& (\sigma^{(1)}_x\otimes \sigma^{(2)}_y + \sigma^{(1)}_y\otimes \sigma^{(2)}_x)|\xi\rangle\nonumber
\end{eqnarray}
commute with the vector field generated by the coherence operator, $C|\xi\rangle = (\sigma^{(1)}_x\otimes \sigma^{(2)}_x + \sigma^{(1)}_y\otimes \sigma^{(2)}_y)|\xi\rangle = |01\rangle\langle10|\xi\rangle$, which implies that $L_{\delta_i} y(t)=0$, or that the vector fields $\delta_i$ are within $\mbox{ker}(dy)$. It can also be seen that the the corresponding Lie brackets $[\delta_i,\delta_j]$ lie within $\Delta = \mbox{span}\{\delta_1,\cdots,\delta_5\}$.

Hence the invariant subspace for the above quantum system is identified to be generated by 5 Hermitian operators in the system's Hilbert space(but not all linearly independent for all the values of the states). This along with three commuting environmental operators($\mathbbm{I},D,D^2$) produces 15 vectors on the analytical manifold which span the invariant subspace for the restructured system. It can also be seen that since the 5 system Hamiltonians do not always generate linearly independent vectors($\delta_1,\cdots, \delta_5$), the rank of the invariant subspace is dependent on the point $\xi$ and hence is singular. The methodology outlined in the proof to construct the control parameters $\alpha(\xi)$ and $\beta(\xi)$ locally around the point $\xi$ works for non-singular invariant distribution $\Delta$ and non-singular control distribution $G$ as well. We can now complete the basis for the tangent space $T_\xi(M)$ with the three commuting vector fields to $\Delta$ which do not belong to $\mbox{ker}(dy)$,
\begin{eqnarray}
d_1 &=& (\sigma^{(1)}_z - \sigma^{(2)}_z)|\xi\rangle\nonumber\\
d_2 &=& (\sigma^{(1)}_x\otimes \sigma^{(2)}_x + \sigma^{(1)}_y\otimes \sigma^{(2)}_y)|\xi\rangle\nonumber\\
d_3 &=& (\sigma^{(1)}_x\otimes \sigma^{(2)}_y - \sigma^{(1)}_y\otimes \sigma^{(2)}_x)|\xi\rangle\nonumber
\end{eqnarray}
The commutation relations are as follows,
\begin{equation*}
\begin{array}{ccc}
&\bullet[\delta_i,\delta_j] \in \Delta; &\bullet [\delta_i,g_j]\in \Delta + G\\
&\bullet[\delta_i,d_j] \in \Delta, d_j\notin \mbox{ker}(dy); &\bullet [d_i, g_j] \in G
\end{array}
\end{equation*}
{\bf Setup:}
\begin{itemize}
\item Let $K = \mbox{rank}\{\delta_1, \cdots , \delta_5\}$ and $\{\delta_1, \cdots, \delta_K\}$ the corresponding vector fields with $\Delta = \mbox{span}\{\delta_1, \cdots, \delta_K\}$.
\item Let $q$ be the minimum number such that $\mbox{rank}\{\Delta, d_1,\cdots,d_3\} = \mbox{rank}\{\Delta, d_1,\cdots,d_q\}$, $q\in{1,2,3}$ and let $\{d_1,\cdots,d_q\}$ be the corresponding linearly independent vector fields with $V_q \triangleq  \mbox{span}\{d_1,\cdots,d_q\}$.
\item Let $r$ be the minimum number such that $\mbox{rank}\{\Delta, V_q, g_1,\cdots,g_{24}\} = \mbox{rank}\{\Delta, V_q, g_1,\cdots,g_r\}$ with $V \triangleq \mbox{span}\{\Delta, V_q, g_1,\cdots,g_r\}$
\end{itemize}
Let the vectors $v_1,\cdots, v_r$ be the linearly independent vectors of $V$ according to the construction above.\\
{\bf The Algorithm:}
\begin{itemize}
\item Solve the equation,
\[
\sum_{j=1}^{24} g_{j}\beta_{ji} = \sum_{k=1}^K c_{ik}v_k + v_{K+i} + \sum_{k=K+q+1}^r c_{ik} v_{K+q+k}
\]
for $i=1,\cdots,q$ with real coefficients $\beta_{ij}$. This is obtained by rewriting the above equation as,
\begin{align*}
[g_1,\cdots,g_{24}, -v_1,\cdots -v_K, -v_{K+q+1}, \cdots, -v_r]\times&&\\
[\beta_{1i},\cdots,\beta_{24i}, c_{i1},\cdots, c_{iK}, c_{K+q+1}, \cdots, c_{ir}]^T&& = v_{K+i}
\end{align*}
A least square solution to the above equation yields the local numerical values of the parameter $\beta_{ji}\in \mathbbm{R}$ for all rows and first $i=1,\cdots q$ columns.
\item Next, solve the equation $\sum_{j=1}^{24} g_j \beta_{ji} = \sum_{k=1}^r c_{ik} v_k$, for $i=q+1,\cdots, 24$ to obtain the rest of the parameters $\beta_{ji}$. This is again obtained by setting,
\begin{align*}
[g_1,\cdots,g_{24}, -v_1,\cdots -v_r]\times [\beta_{1i}, \cdots, \beta_{24i}, c_{i1},\cdots, c_{ir}]^T = 0
\end{align*}
Hence the null space of the matrix $[G, V]$ provides values for the parameters.
\item Finally the control parameters $\alpha$ are obtained by the solution to the equation,
\[
\sum_{j=1}^{24} \alpha_{j} g_j + K_0 = \sum_{k=1}^K c_k v_k + \sum_{k=K+q+1}^r c_k v_k
\]
or the least square solution to the matrix vector equation,
\begin{align*}
[g_1,\cdots,g_{24}, -v_1,\cdots -v_K, -v_{K+q+1}, \cdots, -v_r]\times&&\\
[\alpha_1, \cdots, \alpha_{24}, c_1, \cdots, c_K, c_{K+1+1}, \cdots, c_r]^T&& = -K_0
\end{align*}
\end{itemize}

\end{document}